\documentclass[11pt,a4paper,reqno,intlimits,sumlimits]{amsart}

\usepackage{amssymb}
\usepackage{color}
\usepackage[mathscr]{euscript}
\usepackage{xspace}
\usepackage{enumerate}
\usepackage{epsfig,rotating}
\usepackage{graphicx}
\input{xy}\xyoption{all}
\usepackage{todonotes}
\usepackage{amsmath}
\usepackage{lscape}
\usepackage[latin1]{inputenc}
\usepackage{chicago}

\usepackage[left=1.42in,right=1.42in,top=1.6in,bottom=1.3in]{geometry}

\DeclareMathAlphabet{\mathpzc}{OT1}{pzc}{m}{it}

\usepackage[bookmarksopen,pdfstartview=FitH]{hyperref}
\hypersetup{colorlinks,%
            citecolor=blue,%
            filecolor=blue,%
            linkcolor=blue,%
            urlcolor=blue}%

\begin{document}

%% ========================================================================== %%
%% ========================================================================== %%
\theoremstyle{plain}
\newtheorem{theorem}{Theorem}[section]
\newtheorem{lemma}[theorem]{Lemma}
\newtheorem{proposition}[theorem]{Proposition}
\newtheorem{claim}[theorem]{Claim}
\newtheorem{corollary}[theorem]{Corollary}
\newtheorem{axiom}{Axiom}

\theoremstyle{definition}
\newtheorem{remark}[theorem]{Remark}
\newtheorem{note}{Note}[section]
\newtheorem{definition}[theorem]{Definition}
\newtheorem{example}[theorem]{Example}
\newtheorem*{ackn}{Acknowledgements}
\newtheorem{assumption}{Assumption}
\newtheorem{approach}{Approach}
\newtheorem{critique}{Critique}
\newtheorem{question}{Question}
\newtheorem{aim}{Aim}
\newtheorem*{assucd}{Assumption ($\mathbb{CD}$)}
\newtheorem*{asa}{Assumption ($\mathbb{A}$)}
\newtheorem*{appS}{Approximation ($\mathbb{S}$)}
\newtheorem*{appBS}{Approximation ($\mathbb{BS}$)}
% \newtheorem*{asp}{Assumption ($\mathbb{P}$)}
% \newtheorem*{ass}{Assumption ($\mathbb{S}$)}
%% ========================================================================== %%
%% ========================================================================== %%
\renewcommand{\theequation}{\thesection.\arabic{equation}}
\numberwithin{equation}{section}

\renewcommand{\thefigure}{\thesection.\arabic{figure}}
\numberwithin{equation}{section}

\newcommand{\Law}{\ensuremath{\mathop{\mathrm{Law}}}}
\newcommand{\loc}{{\mathrm{loc}}}

\let\SETMINUS\setminus
\renewcommand{\setminus}{\backslash}

\def\stackrelboth#1#2#3{\mathrel{\mathop{#2}\limits^{#1}_{#3}}}

\makeatletter
\def\Ddots{\mathinner{\mkern1mu\raise\p@
\vbox{\kern7\p@\hbox{.}}\mkern2mu
\raise4\p@\hbox{.}\mkern2mu\raise7\p@\hbox{.}\mkern1mu}}
\makeatother

\newcommand\llambda{{\mathchoice
     {\lambda\mkern-4.5mu{\raisebox{.4ex}{\scriptsize$\backslash$}}}
     {\lambda\mkern-4.83mu{\raisebox{.4ex}{\scriptsize$\backslash$}}}
     {\lambda\mkern-4.5mu{\raisebox{.2ex}
{\footnotesize$\scriptscriptstyle\backslash$}}}
     {\lambda\mkern-5.0mu{\raisebox{.2ex}
{\tiny$\scriptscriptstyle\backslash$}}}}}

\newcommand{\prozess}[1][L]{{\ensuremath{#1=(#1_t)_{t\in[0,T]}}}\xspace}
\newcommand{\prazess}[1][L]{{\ensuremath{#1=(#1_t)_{t\ge0}}}\xspace}
\newcommand{\pt}[1][N]{\ensuremath{\P_{#1}}\xspace}
\newcommand{\tk}[1][N]{\ensuremath{T_{#1}}\xspace}
\newcommand{\dd}[1][]{\ensuremath{\ud{#1}}\xspace}

\newcommand{\scal}[2]{\ensuremath{\langle #1, #2 \rangle}}
\newcommand{\bscal}[2]{\ensuremath{\big\langle #1, #2 \big\rangle}}
\newcommand{\set}[1]{\ensuremath{\left\{#1\right\}}}

%\newcommand{\R}[1][\R]{\ensuremath{\R^{#1}}\xspace}
% \newcommand{\pt}[1][]{{\ensuremath{\P_{T^*_{#1}}}}\xspace}
% \newcommand{\ts}[1][]{\ensuremath{T^*_{#1}}\xspace}
%% ========================================================================== %%
%% ========================================================================== %%
\def\lev{L\'{e}vy\xspace}
\def\lk{L\'{e}vy--Khintchine\xspace}
\def\lib{LIBOR\xspace}
\def\mg{martingale\xspace}
\def\smmg{semimartingale\xspace}
\def\alm{affine LIBOR model\xspace}
\def\alms{affine LIBOR models\xspace}
\def\dalms{defaultable \alms}
\def\ap{affine process\xspace}
\def\aps{affine processes\xspace}

\def\half{\frac{1}{2}}

\def\F{\ensuremath{\mathcal{F}}}
\def\bD{\mathbf{D}}
\def\bF{\mathbf{F}}
\def\bG{\mathbf{G}}
\def\bH{\mathbf{H}}
\def\R{\ensuremath{\mathbb{R}}}
\def\Rp{\mathbb{R}_{\geqslant0}}
\def\Rm{\mathbb{R}_{\leqslant 0}}
\def\C{\ensuremath{\mathbb{C}}}
\def\U{\ensuremath{\mathcal{U}}}
\def\I{\mathcal{I}}
\def\N{\mathbb{N}}

\newcommand{\la}{\langle}
\newcommand{\ra}{\rangle}

\def\P{\ensuremath{\mathrm{I\kern-.2em P}}}
\def\Q{\mathbb{Q}}
\def\E{\ensuremath{\mathrm{I\kern-.2em E}}}

\def\hP{\ensuremath{\widehat{\mathrm{I\kern-.2em P}}}}
\def\hE{\ensuremath{\widehat{\mathrm{I\kern-.2em E}}}}

\def\bP{\ensuremath{\overline{\mathrm{I\kern-.2em P}}}}
\def\bE{\ensuremath{\overline{\mathrm{I\kern-.2em E}}}}

\def\bphi{\overline{\phi}}
\def\bpsi{\overline{\psi}}

\def\ott{{0\leq t\leq T}}
\def\idd{{1\le i\le d}}

\def\icc{\mathpzc{i}}
\def\ecc{\mathbf{e}_\mathpzc{i}}

\def\uk{u_{k+1}}
\def\vk{v_{k+1}}

\def\e{\mathrm{e}}
\def\ud{\ensuremath{\mathrm{d}}}
\def\dt{\ud t}
\def\ds{\ud s}
\def\dx{\ud x}
\def\dy{\ud y}
\def\dv{\ud v}
\def\dw{\ud w}
\def\dz{\ud z}
\def\dsdx{\ensuremath{(\ud s, \ud x)}}
\def\dtdx{\ensuremath{(\ud t, \ud x)}}

\def\lsnc{\ensuremath{\mathrm{LSNC-}\chi^2}}
\def\nc{\ensuremath{\mathrm{NC-}\chi^2}}

\def\red{\color{red}}
\def\blue{\color{blue}}
%% ========================================================================== %%
%% ========================================================================== %%
\newcommand{\cD}{{\mathcal{D}}}
\newcommand{\cF}{{\mathcal{F}}}
\newcommand{\cG}{{\mathcal{G}}}
\newcommand{\cH}{{\mathcal{H}}}
\newcommand{\cK}{{\mathcal{K}}}
\newcommand{\cM}{{\mathcal{M}}}
\newcommand{\cT}{{\mathcal{T}}}
\newcommand{\ha}{{\mathbb{H}}}
\newcommand{\indik}{{\mathbf{1}}}
\newcommand{\ifdefault}[1]{\ensuremath{\mathbf{1}_{\{\tau \leq #1\}}}}
\newcommand{\ifnodefault}[1]{\ensuremath{\mathbf{1}_{\{\tau > #1\}}}}

\newcommand\bovermat[2]{%
  \makebox[0pt][l]{$\smash{\overbrace{\phantom{%
    \begin{matrix}#2\end{matrix}}}^{#1}}$}#2}
\newcommand\bundermat[2]{%
  \makebox[0pt][l]{$\smash{\underbrace{\phantom{%
    \begin{matrix}#2\end{matrix}}}_{#1}}$}#2}
\makeatletter

%% ========================================================================== %%
%% ========================================================================== %%

\title[Approximate option pricing in the Lévy Libor model ]
    {Approximate option pricing in the Lévy Libor model}

\author[Z. Grbac]{Zorana Grbac}
\address{Z. Grbac - Laboratoire de Probabilit\'es et Mod\`eles Al\'eatoires, Universit\'e Paris Diderot, France}
\email{grbac@math.univ-paris-diderot.fr}

\author[D. Krief]{David Krief}
\address{D. Krief - Laboratoire de Probabilit\'es et Mod\`eles Al\'eatoires, Universit\'e Paris Diderot, France}
\email{krief@math.univ-paris-diderot.fr}

\author[P. Tankov]{Peter Tankov}
\address{P. Tankov - Laboratoire de Probabilit\'es et Mod\`eles Al\'eatoires, Universit\'e Paris Diderot, France}
\email{peter.tankov@polytechnique.org}

%\thanks{}

%\keywords{}
%\subjclass[2010]{}
% \newline\indent\emph{JEL classification.} }

\begin{abstract}
In this paper we consider the pricing of options on interest rates such as caplets
and swaptions in the Lévy Libor model developed by \citeN{EberleinOezkan05}. This model is an extension to Lévy driving processes of the
classical log-normal Libor market model (LMM) driven by a Brownian motion. Option
pricing is significantly less tractable in this model than in the LMM
due to the appearance of stochastic terms in the jump part of the
driving process when performing the measure changes which are standard
in pricing of interest rate derivatives. To obtain explicit
approximation for option prices, we propose to
treat a given Lévy Libor model as a suitable perturbation of the log-normal LMM. The method is inspired by recent works by \citeN{CernyDenklKallsen13}  and \citeN{MenasseTankov15}. The approximate option
prices in the Lévy Libor model are given as the corresponding LMM
prices plus correction terms which depend on the characteristics of
the underlying Lévy process and some additional terms obtained from
the LMM model.\\

\noindent Key words: Libor market model, caplet, swaption, Lévy Libor model, asymptotic approximation.\\

\end{abstract}

\date{\today}\maketitle\pagestyle{myheadings}\frenchspacing

\section{Introduction}

The goal of this paper is to develop explicit approximations for
option prices in the Lévy Libor model introduced by \citeN{EberleinOezkan05}. In particular, we shall be interested in price approximations for caplets, whose pay-off is a function of only one underlying Libor rate and swaptions, which can be regarded as options on a ``basket'' of multiple Libor rates of different maturities.

A full-fledged model of Libor rates such as the Lévy Libor model is
typically used for the purposes of pricing and risk management of
exotic interest rate products. The prices and hedge ratios must be
consistent with the market-quoted prices of liquid options, which
means that the model must be calibrated to the available prices /
implied volatilities of caplets and swaptions. To perform such a
calibration efficiently, one therefore needs explicit formulas or fast
numerical algorithms for caplet and swaption prices. 

Computation of option prices in the Lévy Libor model to arbitrary
precision is only possible via Monte Carlo. Efficient simulation
algorithms suitable for pricing exotic options have been proposed in \cite{kohatsu2010jump,papapantoleon12efficient},
however, these Monte Carlo algorithms are probably not an option for the purposes
of calibration because the computation is still too slow due to the
presence of both discretization and statistical error.

\citeN{EberleinOezkan05}, \citeN{Kluge05} and
\cite{belomestny2011jump} propose fast methods for
computing caplet prices which are based on Fourier transform
inversion and use the fact that the characteristic function of many
parametric Lévy processes is known explicitly. Since in the Lévy Libor
model, the Libor rate $L^k$ is not a geometric Lévy process under
the corresponding probability measure $\mathbb Q^{T_k}$, unless $k=n$ (see
Remark \ref{nolevy.rq} below for details), using
these methods for $k<n$ requires an additional approximation (some
random terms appearing in the compensator of the jump measure of $L^k$
are approximated by their values at time $t=0$, a method known as
freezing).

In this paper we take an alternative route and develop approximate
formulas for caplets and swaptions using asymptotic expansion
techniques. 
Inspired by methods used in \citeN{CernyDenklKallsen13} and
\citeN{MenasseTankov15} (see also
\cite{benhamou2009smart,benhamou2010time} for related expansions
``around a Black-Scholes proxy'' in other models), we consider a given Lévy Libor model as a
perturbation of the log-normal LMM. Starting from the
driving Lévy process $(X_t)_{t\geq 0}$ of the Lévy Libor model,
assumed to have zero expectation, we
introduce a family of processes $X^\alpha_t  = \alpha
X_{t/\alpha^2}$ parameterized by $\alpha\in (0,1]$, together with the
corresponding family of Lévy Libor models. For $\alpha=1$ one recovers
the original Lévy Libor model. When
$\alpha\to 0$, the family $X^\alpha$ converges weakly in Skorokhod
topology to a Brownian motion, and the option prices in the Lévy Libor model corresponding
to the process $X^\alpha$ converge to the prices in the log-normal LMM. The option prices in
the original Lévy Libor model can then be approximated by their
second-order expansions in the parameter $\alpha$, around the value $\alpha=0$. This leads to an asymptotic approximation formula for a derivative price expressed as a linear combination of the derivative price stemming from the LMM and correction terms depending on the characteristics of the driving Lévy process.
The terms of this expansion are often much easier to
compute than the option prices in the Lévy Libor model. In particular,
we shall see the expansion for caplets is expressed in terms of the derivatives of
the standard Black's formula, and the various terms of the
expansion for swaptions can be approximated using one of the many swaption
approximations for the log-normal LMM available in the literature. 

 This paper is structured as follows. In Section \ref{presentation} we
 briefly review the Lévy Libor model. %and recall the existing methods for computing option prices in this model. 
  In Section \ref{pde} we
 show how the prices of European-style options may be expressed as
 solutions of partial integro-differential equations (PIDE). These
 PIDEs form the basis of our asymptotic method, presented in detail in
 Section \ref{asymptotics}. Finally, numerical illustrations are
 provided in Section \ref{numerics}. 
\section{Presentation of the model}\label{presentation}

In this section we present a slight modification of the \lev Libor model by \citeN{EberleinOezkan05}, which is a generalization, based on \lev processes, of the Libor market model driven by a Brownian motion, introduced by \shortciteN{SandmannSondermannMiltersen95}, \shortciteN{BraceGatarekMusiela97} and \shortciteN{MiltersenSandmannSondermann97}.

Let a discrete tenor structure $0  \leq  T_{0} < T_{1} < \ldots < T_{n}$ be given, and set $\delta_k:= T_{k} -T_{k-1}$, for $k=1, \ldots, n$. We assume that zero-coupon bonds with maturities $T_{k}$, $k=0, \ldots, n$, are traded in the market. The time-$t$ price of a bond with maturity $T_k$ is denoted by $B_t(T_{k})$ with $B_{T_{k}}(T_{k})=1$.

For every tenor date $T_{k}$, $k=1, \ldots, n$, the forward Libor rate $L_t^k$ at time $t \leq T_{k-1}$ for the accrual period $[T_{k-1}, T_{k}]$ is a discretely compounded interest rate defined as
\begin{equation}
\label{default-free-libor-rate}
L_t^k:= \frac{1}{\delta_{k}}
  \left( \frac{B_t(T_{k-1})}{B_t(T_{k})} - 1\right).
\end{equation}
For all $t > T_{k-1}$, we set $L_t^k:= L_{T_{k-1}}^k$.

To set up the Libor model, one needs to specify the forward Libor rates $L_t^k$, $k=1, \ldots, n$, such that each Libor rate $L^k$ is a martingale with respect to the corresponding forward measure $\Q^{T_{k}}$ using the bond with maturity $T_k$ as numéraire.  We recall that the forward measures are interconnected via the Libor rates themselves and hence each Libor rate depends also on some other Libor rates as we shall see below. More precisely, assuming that the forward measure $\Q^{T_{n}}$ for the most distant maturity $T_n$ (i.e. with numéraire $B(T_n)$) is given, the link between the forward measure $\Q^{T_{k}}$ and $\Q^{T_n}$ is provided by 
\begin{align}
\label{eq:forward-measure}
\frac{d \Q^{T_k}}{d \Q^{T_n}} \Big|_{\mathcal{F}_t} &  = \frac{B_t(T_k)}{B_t(T_n)} \frac{B_0(T_n)}{B_0(T_k)}  = \prod_{j=k+1}^{n} \frac{1 + \delta_j L^j_t}{1 + \delta_j L^j_0},
\end{align}
for every $k=1, \ldots, n-1$. The forward measure $\Q^{T_{n}}$ is referred to as the terminal forward measure.

\subsection{The driving process}

Let us denote by $(\Omega, \cF, \bF=(\cF_{t})_{0\leq t \leq T^{*}},
\Q^{T_n})$ a complete stochastic basis and let $X$ be an $\R^d$-valued
Lévy process $(X_t)_{0 \leq t \leq T^{*}}$ on this stochastic basis
with Lévy measure $F$ and diffusion matrix $c$.   The filtration $\bF$ is generated by $X$ and
 $\Q^{T_n}$ is the forward measure associated with the date $T_n$,
 i.e. with the numeraire $B_t(T_n)$. The process $X$ is assumed
 without loss of generality to be driftless under $\Q^{T_n}$.

 Moreover, we assume that $\int_{|z| >1} |z| F(dz) < \infty$. This implies in addition that $X$ is a special semimartingale and allows to choose the truncation function $h(z)=z$, for $z \in \R^d$. The canonical representation of $X$ is given by 
 
\begin{equation}
\label{eq:driving-levy}
X_{t} =  \sqrt{c}  W_{t}^{T_{n}} + \int_{0}^{t} \int_{\R^d} z (\mu - \nu^{T_{n}})(ds, dz),
\end{equation}
where $W^{T_{n}}=(W^{T_{n}}_t)_{0 \leq t \leq T_n}$ denotes a standard $d$-dimensional Brownian motion with respect to the measure $\Q^{T_{n}}$, $\mu$ is the random measure of jumps of $X$ and $\nu^{T_{n}}(ds, dz)= F(dz) ds$ is the $\Q^{T_{n}}$-compensator of $\mu$.

\subsection{The model}

Denote by $L= (L^1, \ldots, L^n)^{\top}$ the column vector of  forward Libor rates. We assume that under the terminal measure $\Q^{T_n}$, the dynamics of $L$ is given by the following SDE 
\begin{equation}
\label{eq:Libor-terminal-measure}
d L_t = L_{t-} (b(t, L_{t}) dt + \Lambda(t) d X_t),
\end{equation}
where $b(t, L_{t})$ is the drift term and $\Lambda(t)$ a deterministic $n\times d$ volatility matrix. We write $\Lambda(t) = (\lambda^1(t), \ldots, \lambda^n(t))^{\top}$, where $\lambda^k(t)$ denotes the $d$-dimensional volatility vector of the Libor rate $L^k$ and  assuming that $\lambda^k(t) =0 $, for $t> T_{k-1}$. 

One typically assumes that the jumps of $X$ are bounded from below,
i.e. $\Delta X_t > C$, for all $t\in [0, T^*]$ and for some strictly
negative constant $C$, which is chosen such that it ensures the positivity of the Libor rates given by \eqref{eq:Libor-terminal-measure}.

The drift $b(t, L_{t}) = (b^{1}(t, L_{t}), \ldots, b^{n}(t, L_{t}) )$ is determined by the no-arbitrage requirement that $L^k$ has to be a martingale with respect to $\Q^{T_k}$, for every $k=1, \ldots, n$. This yields
\begin{align}
\label{eq:drift}
b^{k}(t, L_t) & = -  \sum_{j=k+1}^{n} \frac{\delta_j L_{t}^j}{1+ \delta_j L_{t}^j} \langle \lambda^k(t), c \, \lambda^j(t) \rangle \\
\notag & + \int_{\R^d} \langle \lambda^k(t), z \rangle \left( 1 - \prod_{j=k+1}^{n} \left( 1 + \frac{\delta_j L_{t}^j \langle \lambda^j(t), z \rangle} {1+ \delta_j L_{t}^j} \right)\right) F(dz).
\end{align}
The above drift condition follows from \eqref{eq:forward-measure} and Girsanov's theorem for semimartingales noticing that 
\begin{align*}
d L_t^k& = L_{t-}^k(b^{k}(t, L_{t}) dt + \lambda^k(t) d X_t) \\
& = L_{t-}^k \lambda^k(t) d X^{T_k}_t,
\end{align*}
where 
\begin{equation}
\label{eq:X-Tk}
X^{T_k}_t = \sqrt{c}W_t^{T_k} + \int_0^t \int_{\R^d} z (\mu - \nu^{T_k}) (ds, dz)
\end{equation}
is a special semimartingale with a $d$-dimensional $\Q^{T_{k}}$-Brownian motion $W^{T_k}$  given by
\begin{equation}
\label{TkBrownian}
d W_{t}^{T_{k}} := d W_{t}^{T_{n}}- \sqrt{c} \left( \sum_{j=k+1}^{n}\frac{\delta_j L_{t}^j}{1+ \delta_j L_{t}^j}   \lambda^j(t) \right) dt
\end{equation}
and the $\Q^{T_{k}}$-compensator $\nu^{T_k}$ of $\mu$ given by 
\begin{eqnarray}
\label{Tkcompensator}
\nu^{T_{k}}(dt, dz)& := &\prod_{j=k+1}^{n} \left( 1 + \frac{\delta_j L_{t-}^j }{1+ \delta_j L_{t-}^j}  \langle \lambda^j(t), z \rangle \right) \nu^{T_{n}} (dt, dz) \\
\nonumber & = & \prod_{j=k+1}^{n} \left( 1 + \frac{\delta_j L_{t}^j }{1+ \delta_j L_{t}^j}  \langle \lambda^j(t), z \rangle \right) F(dz) dt \\ 
\nonumber & = &  F^{T_k}_t(dz) dt 
\end{eqnarray}
with 
\begin{equation}
\label{Tk-Levy-measure}
F^{T_k}_t(dz) := \prod_{j=k+1}^{n} \left( 1 + \frac{\delta_j L_{t}^j }{1+ \delta_j L_{t}^j}  \langle \lambda^j(t), z \rangle \right) F(dz).
\end{equation}

Equalities \eqref{TkBrownian}  and \eqref{Tkcompensator}, and consequently also the drift condition  \eqref{eq:drift}, are implied by Girsanov's theorem for semimartingales applied first to the measure change from $\Q^{T_n}$ to $\Q^{T_{n-1}}$ and then proceeding backwards. We refer to \citeN[Proposition 2.6]{Kallsen06} for a version of Girsanov's theorem that can be directly applied in this case.  Note that the random terms $\frac{\delta_j L_{t}^j}{1+ \delta_j L_{t}^j} $ appear in the measure change due to the fact that for each $j=n, n-1, \ldots, 1$ we have  
\begin{equation}
\label{eq:aux}
d (1+\delta_j L_t^j) = (1+ \delta_j L_{t-}^j) \left(\frac{\delta_j L_{t-}^j}{1+ \delta_j L_{t-}^j} b^{j}(t, L_{t}) dt +  \frac{\delta_j L_{t-}^j}{1+ \delta_j L_{t-}^j} \lambda^j(t) d X_t\right),
\end{equation}
We point out that the predictable random terms $\frac{\delta_j L_{t-}^j}{1+ \delta_j L_{t-}^j} $ can be replaced with $\frac{\delta_j L_{t}^j}{1+ \delta_j L_{t}^j} $ in equalities  \eqref{eq:drift}, \eqref{TkBrownian} and \eqref{Tkcompensator}  due to absolute continuity of the characteristics of $X$.

Therefore, the vector process  of Libor rates $L$, given in \eqref{eq:Libor-terminal-measure}

with the drift \eqref{eq:drift}, is a time-inhomogeneous Markov
process and its infinitesimal generator under $\Q^{T_n}$ is given by
\begin{align}
\label{eq:generator}
\mathcal{A}_t f(x) & = \sum_{i=1}^{n} x_i b^i(t, x) \frac{\partial
  f(x)}{\partial x_i} + \frac{1}{2}  \sum_{i, j=1}^{n}x_i x_j (\Lambda(t) c \Lambda(t)^\top )_{ij} \frac{\partial f(x)}{\partial x_i \partial x_j} \\
\notag & + \int_{\R^d} \left( f(\text{diag}(x)(\mathbf{1}+\Lambda(t)
  z) ) - f(x) - \sum_{j=1}^{n} x_j (\Lambda(t) z)_j  \frac{\partial f(x)}{\partial x_j} \right) F(dz),
\end{align}
for a function $f \in C^2_0(\R^n, \R)$ and with the function $b^i(t, x)$, for $i=1, \ldots, n$ and $x=(x_1, \ldots, x_n) \in \R^n$, given by
\begin{align*}
b^{i}(t, x) & = -  \sum_{j=i+1}^{n} \frac{\delta_j x_j}{1+ \delta_j x_j} \langle \lambda^i(t), c \, \lambda^j(t) \rangle \\
\notag & + \int_{\R^d} \langle \lambda^i(t), z \rangle \left( 1 - \prod_{j=k+1}^{n} \left( 1 + \frac{\delta_j x_j \langle \lambda^j(t), z \rangle} {1+ \delta_j x_j} \right)\right) F(dz). 
\end{align*}

\begin{remark}[Connection to the Lévy Libor model of \citeN{EberleinOezkan05}]

The dynamics of the forward Libor rate $L^k$, for all $k=1, \ldots, n$, in the Lévy Libor model of \citeN{EberleinOezkan05} (compare also %\citeN{Eberlein2006c}
\citeN{EberleinKluge07}) is given as an ordinary exponential of the following form 
\begin{align}
L_t^k&  = L_0^k \exp \left( \int_0^t \tilde b^k(s, L_s) ds + \int_0^t \tilde \lambda^k(s) d \tilde Y_s \right),
\end{align}
for some deterministic volatility vector $\tilde \lambda^k$ and the drift $\tilde b^k(t, L_t) $ which has to be chosen such that the Libor rate $L^k$ is a martingale under the forward measure $\Q^{T_k}$. Here $\tilde Y$ is a $d$-dimensional Lévy process given by 
$$
\tilde Y_{t} =  \sqrt{c}  W_{t}^{T_{n}} + \int_{0}^{t} \int_{\R^d} z (\tilde \mu - \tilde \nu^{T_{n}})(ds, dz),
$$
with the $\Q^{T_n}$-characteristics $(0, c, \tilde F)$, where $\tilde \nu^{T_{n}}(ds, dz) = \tilde F(dz) ds$. The Lévy measure  $\tilde F$ has to satisfy the usual integrability conditions  ensuring the finiteness of the exponential moments. The dynamics of $L^k$ is thus given by the following SDE 
\begin{align*}
d L_t^k&  = L_{t-}^k \left(  b^k(t, L_t) dt +   \sqrt{c}  \tilde \lambda^k(t)  d W_{t}^{T_{n}}  + (e^{\langle \tilde  \lambda^k(t), z\rangle } -1 ) (\tilde \mu - \tilde \nu^{T_n}) (dt, dz)  \right) \\
& = L_{t-}^k \left(  b^k(t, L_t) dt  + d Y^k_t \right) ,
\end{align*}
for all $k$, where  $Y^k$ is a time-inhomogeneous Lévy process given by 
$$
Y^k_{t} =  \int_0^t \sqrt{c}  \tilde \lambda^k(s) d W_{s}^{T_{n}} + \int_{0}^{t} \int_{\R^d} (e^{\langle \tilde  \lambda^k(s), z\rangle } -1 ) (\tilde \mu - \tilde \nu^{T_{n}})(ds, dz)
$$
and the drift $b^k(t, L_t)$ is given by 
\begin{align*}
b^k(t, L_t) & = \tilde b^k(t, L_t) + \frac{1}{2} \langle \tilde \lambda^k(t), c \tilde \lambda^k(t) \rangle \\
& + \int_{\R^d} (e^{\langle \tilde  \lambda^k(t), z\rangle } -1 -  \langle \tilde  \lambda^k(t), z\rangle ) \tilde F ( dz). 
\end{align*}

\end{remark}

\section{Option pricing via PIDEs}
\label{pde}
Below we present the pricing PIDEs related to general option payoffs and then more specifically to caplets and swaptions. We price all options under the given terminal measure $\Q^{T_n}$.
\subsection{General payoff} Consider a European-type payoff with maturity $T_k$ given by $\xi=g(L_{T_k})$, for some tenor date $T_k$. Its time-$t$ price $P_t$ is given by the following risk-neutral pricing formula 
 \begin{align*}
 P_t & = B_t(T_k) \E^{\Q^{T_k}} [g(L_{T_k}) \mid \mathcal{F}_t]  \\
 & = B_t(T_n) \E^{\Q^{T_n}} \left[ \frac{B_{T_k}(T_k)}{B_{T}(T_n)} g(L_{T_k}) \mid \mathcal{F}_t\right]  \\
 & = B_t(T_n) \E^{\Q^{T_n}} \left[ \prod_{j=k+1}^{n} (1 + \delta_j L^j_{T_k}) g(L_{T_k}) \mid \mathcal{F}_t\right] \\
 & = B_t(T_n) u(t, L_t),
   \end{align*}
where $u$ is the solution of the following PIDE\footnote{A detailed
  proof of this statement is out of scope of this note. Here we simply
assume that Equation \eqref{eq:general-PIDE} admits a unique solution
which is sufficiently regular and is of polynomial growth. The
existence of such a solution may be established first by Fourier
methods for the case when there is no drift and then by a fixed-point
theorem in Sobolev spaces using the regularizing properties of the
Lévy kernel for the general case (see \cite[Chapter 7]{defranco.thesis} for
similar arguments). Once the existence of a regular solution has been
established, the expression for the option price follows by the standard
Feynman-Kac formula.}
 \begin{align}
 \label{eq:general-PIDE}
 \partial_t u + \mathcal{A}_t u  & = 0 \\
  \notag u(T_k, x) &   = \tilde g(x)
   \end{align}
and $\tilde g$ denotes the transformed payoff function given by 
$$
\tilde g(x) := \tilde g(x_1, \ldots, x_n) = \prod_{j=k+1}^{n} (1 + \delta_j x_j) g(x_1, \ldots, x_n).
$$ 

In what follows we shall in particular focus on two most liquid interest rate options: caps (caplets) and swaptions. 

\subsection{Caplet} Consider a caplet with strike $K$ and payoff $\xi=\delta_k (L^k_{T_{k-1}} - K)^+$ at time $T_{k}$. Note that here the payoff is in fact a $\mathcal{F}_{T_{k-1}}$-measurable random variable and it is paid at time $T_{k}$. This is known as \textit{payment in arrears}. There exist also other conventions for caplet payoffs, but this one is the one typically used. 

The time-$t$ price of the caplet, denoted by $P_t^{Cpl}$ is thus given by 
 \begin{align}
 \label{eq:caplet-price}
 P_t^{Cpl}  & = B_t(T_{k}) \delta_{k} \E^{\Q^{T_{k}}} [(L^k_{T_{k-1}} - K)^+ \mid \mathcal{F}_t] \\
 \notag & = B_t(T_n) \delta_k \E^{\Q^{T_n}} \left[ \prod_{j=k+1}^{n} (1 + \delta_j L^j_{T_{k-1}}) (L^k_{T_{k-1}} - K)^+ \mid \mathcal{F}_t\right] \\
 \notag & = B_t(T_n) \delta_k u(t, L_t)
  \end{align}
where $u$ is the solution to 
 \begin{align}
  \label{eq:caplet-PIDE}
 \partial_t u + \mathcal{A}_t u  & = 0 \\
\notag  u(T_{k-1}, x) & = \tilde g(x)
   \end{align}
with  
$$\tilde g(x) :=  (x_k - K)^+ \prod_{j=k+1}^{n} (1 + \delta_j x_j).
$$
 For the second equality in \eqref{eq:caplet-price} we have used the measure change from $\Q^{T_k}$ to $\Q^{T_n}$ given in \eqref{eq:forward-measure}.

\begin{remark}\label{nolevy.rq}
Noting that the payoff of the caplet depends on one single underlying forward Libor rate $L^k$, it is often more convenient to price it directly under the corresponding forward measure $\Q^{T_k}$, using the first equality in \eqref{eq:caplet-price}. Thus, one has 
$$
P_t^{Cpl} = B_t(T_k)\delta_k  u(t, L_t),
$$
where $u$ is the solution to 
 \begin{align}
  \label{eq:caplet-PIDE-measure-Tk}
 \partial_t u + \mathcal{A}^{T_k}_t u  & = 0 \\
\notag  u(T_{k-1}, x) & = \tilde g(x)
   \end{align}
with  $\tilde g(x) := (x_k - K)^+$ and where $\mathcal{A}^{T_k}$ is the generator of $L$ under the forward measure $\Q^{T_k}$.  %given by \eqref{eq:Tk-generator}.
In the log-normal LMM this leads directly to the Black's formula for caplet prices. However, in the Lévy Libor model the driving process $X$ under the forward measure  $\Q^{T_k}$ is not a Lévy process anymore since its compensator of the random measure of jumps becomes stochastic (see \eqref{Tk-Levy-measure}). Therefore, passing to the forward measure in this case does not lead to a closed-form pricing formula and does not bring any particular advantage. This is why in the forthcoming section we shall work directly under the terminal measure $\Q^{T_n}$.
\end{remark}

\subsection{Swaptions}
\label{s:swaptions}
 Let us consider a swaption, written on a fixed-for-floating (payer) interest rate swap with inception date $T_0$, payment dates  $T_1, \ldots, T_n$ and nominal $N=1$. We denote by $K$ the swaption strike rate and assume for simplicity that the maturity $T$ of the swaption coincides with the inception date of the underlying swap, i.e. we assume $T=T_0$. Therefore, the payoff of the swaption at maturity is given by $\left(P^{Sw}(T_0; T_0, T_n, K) \right)^{+} $, where $P^{Sw}(T_0; T_0, T_n, K)$ denotes the value of the swap with fixed rate $K$ at time $T_0$ given by 
\begin{align*}
P^{Sw}(T_0; T_0, T_n, K)  & =  \sum_{j=1}^{n} \delta_{j} B_{T_0}(T_j) \E^{\Q^{T_j}} \left[ L^j_{T_{j-1}} - K | \mathcal{F}_{T_0} \right] \\
& =  \sum_{j=1}^{n} \delta_{j}   B_{T_0}(T_j) \left( L^j_{T_{0}} - K \right) \\
& =  (\sum_{j=1}^{n} \delta_{j}   B_{T_0}(T_j)) \left( R(T_0; T_0, T_n) - K \right) 
\end{align*}
where
\begin{align}
\label{eq:swap-rate}
R(t; T_0, T_n) & =  \frac{\sum_{j=1}^{n} \delta_{j} B_t(T_j) L^j_{t}}{\sum_{j=1}^{n} \delta_{j} B_t(T_j)}=:  \sum_{j=1}^{n} w_{j}  L^j_{t}
\end{align}
is the swap rate i.e. the fixed rate such that the time-$t$ price of the swap is equal to zero. 
Here we denote 
\begin{align}
\label{eq:swap-weights}
w_{j}(t)& :=  \frac{\delta_{j} B_t(T_j)}{\sum_{k=1}^{n} \delta_{k} B_t(T_k)} % = \frac{\delta_{j} B_t(T_j)}{N_t} 
\end{align}

Note that $\sum_{j=1}^{n} w_{j}(t) =1$. Dividing the numerator and the denominator in \eqref{eq:swap-rate} by $B_t(T_n)$ and using the telescopic products together with  \eqref{default-free-libor-rate} we see that $w_j(t) = f_j(L_t)$ for a function $f_j$ given by
\begin{align}
\label{eq:swap-weights-Libors}
f_j(x) & = \frac{\delta_{j} \prod_{i=j+1}^{n} (1+ \delta_i x_i)}{\sum_{k=1}^{n} \delta_{k} \prod_{i=k+1}^{n} (1+ \delta_i x_i)}
\end{align}
for $j=1, \ldots, n$. 

Therefore, the swaption price at time $t \leq T_0$ is given by
\begin{align}
\label{eq:swaption-price}
\notag & P^{Swn}(t; T_0, T_n, K) \\
 & \quad =  B_{t}(T_0) \E^{\Q^{T_0}}  \left[ \left( P^{Sw}(T_0;  T_0, T_n, K) \right)^{+} | \mathcal{F}_t  \right]\\
  \notag & \quad  = B_{t}(T_0) \E^{\Q^{T_0}}  \left[    (\sum_{j=1}^{n} \delta_{j}   B_{T_0}(T_j)) \left( R(t_0; T_0, T_n) - K  \right)^{+}  | \mathcal{F}_t  \right]\\ 
  \notag & \quad=  B_t(T_n) \E^{\Q^{T_n}} \left[   \frac{\sum_{j=1}^{n} \delta_{j}   B_{T_0}(T_j)}{B_{T_0}(T_n)} \left( R(t_0; T_0, T_n) - K \right)^{+}  | \mathcal{F}_t  \right] \\
\notag  & \quad = B_t(T_n) u(t, L_t)
\end{align}
where $u$ is the solution to 
 \begin{align}
  \label{eq:swaption-PIDE-1}
 \partial_t u + \mathcal{A}_t u  & = 0 \\
\notag  u(T_0, x) & = \tilde g(x)
   \end{align}
with  $\tilde g(x) :=   \delta_n f_n(x)^{-1} \left( \sum_{j=1}^{n} f_j(x) x_j - K\right)^+$.

\section{Approximate pricing}
\label{asymptotics}

\subsection{Approximate pricing for general payoffs under the terminal measure}
\label{asympt-general-payoff}
Following an approach introduced by  \citeN{CernyDenklKallsen13}, we
introduce a small parameter into the model by defining the rescaled
L\'evy process $X^\alpha_t:= \alpha X_{t/\alpha^2}$ with $\alpha\in
(0,1)$. The process $X^\alpha$ is a martingale L\'evy process  under the terminal measure $\Q^{T_n}$ with
characteristic triplet $(0,c,F_\alpha)$ with respect to the truncation
function $h(z)=z$, where 
$$
F_\alpha(A) = \frac{1}{\alpha^2} F(\{z\in \mathbb R^d: z\alpha \in
A\},\quad \text{for}\ A\in \mathcal B(\mathbb R^d). 
$$
We now consider a family of L\'evy Libor models driven by the processes
$X^\alpha$, $\alpha \in (0, 1)$, and defined by 
\begin{align}
\label{eq:alpha-Libor}
dL^\alpha_t = L^\alpha_{t-}(b_\alpha(t,L^\alpha_t) dt + \Lambda(t)dX^\alpha_t),
\end{align}
where the drift $b_\alpha$ is given by \eqref{eq:drift} with $F$
replaced by $F_\alpha$. Substituting the explicit form of $F_\alpha$,
we obtain
\begin{align*}
b_\alpha^{k}(t, L_t) & = -  \sum_{j=k+1}^{n} \frac{\delta_j L_{t}^j}{1+ \delta_j L_{t}^j} \langle \lambda^k(t), c \, \lambda^j(t) \rangle \\
\notag & + \frac{1}{\alpha}\int_{\R^d} \langle \lambda^k(t), z \rangle
\left( 1 - \prod_{j=k+1}^{n} \left( 1 + \frac{\alpha\delta_j L_{t}^j
      \langle \lambda^j(t), z \rangle} {1+ \delta_j L_{t}^j}
  \right)\right) F(dz)\\
& = -  \sum_{j_0 = k+1}^n \Sigma_{k j_0}(t) \frac{\delta_{j_0} L^{j_0}_t}{1+\delta_{j_0} L^{j_0}_t}   \\
&  - \sum_{p=1}^{n-k-1} \alpha^p  \sum_{j_0 = k+1}^n
\sum_{j_1 = j_0+1}^n \dots \sum_{j_p = j_{p-1}+1}^n M^{p+2}_t(\lambda^{k}, \lambda^{j_0}, \ldots,\lambda^{j_p}) \prod_{l=0}^p
\frac{\delta_{j_l} L^{j_l}_t}{1+\delta_{j_l} L^{j_l}_t} \\
& =: - \sum_{p=0}^{n-k-1} \alpha^p b_p^{k}(t, L_t)
\end{align*}
where we define
\begin{equation}
\label{eq:BS-vol}
\Sigma_{ij}(t) := (\Lambda(t) c
\Lambda(t)^\top )_{ij}  + \int_{\mathbb R^d} \langle \lambda^i(t), z\rangle \langle \lambda^j(t),
z\rangle F(dz),
\end{equation}
for all $i, j =1, \ldots, n$, and  
\begin{equation}
\label{eq:moments-Levy-measure}
M^k_t(\lambda^{1},\dots,\lambda^k) :=  \int_{\mathbb
  R^d} \prod_{p=1}^k \langle \lambda^{p}(t), z\rangle F(dz)
\end{equation}
for all $k=1, \ldots, n$. We denote the infinitesimal generator of $L^\alpha$ by $\mathcal
A^\alpha_t$. 
For a smooth function $f:\mathbb R^d \to \mathbb R$, the infinitesimal
generator $\mathcal A^\alpha_tf$ can be expanded in powers of $\alpha$ as follows:
\begin{multline*}
\mathcal A^\alpha_t f(x) = \sum_{i=1}^n b^i_\alpha(t,x) x_i\frac{\partial
  f(x)}{\partial x_i} + \frac{1}{2}\sum_{i,j=1}^n \Sigma_{i j }(t) x_i x_j \frac{\partial^2 f(x)}{\partial
  x_i \partial x_j}  \\+ \sum_{k=3}^\infty \sum_{i_1,\dots,i_k =
  1}^{ n}\frac{\alpha^{k-2}}{k!}x_{i_1}\dots x_{i_k}\frac{\partial^k
  f(x)}{\partial x_{i_1} \dots \partial x_{i_k}} M^k_t(\lambda^{i_1},\dots,\lambda^{i_k}).
\end{multline*}
Consider now a financial product whose price is given by a generic
PIDE of the form \eqref{eq:general-PIDE} with $\mathcal A_t$ replaced
by $\mathcal A^\alpha_t$. Assuming sufficient
regularity\footnote{See \cite{MenasseTankov15} for rigorous arguments in a simplified but similar setting.}, one may expand the solution $u^\alpha$ in powers of $\alpha$:
\begin{align}
\label{eq:expansion-u}
u^\alpha(t,x) = \sum_{p=0}^\infty \alpha^p u_p (t,x). 
\end{align}
Substituting the expansions for $\mathcal A^\alpha_t$ and $b_\alpha$
into this equation, and gathering terms with the same power of
$\alpha$, we obtain an 'open-ended' system of PIDE for the terms in the
expansion of $u^\alpha$. 

The zero-order term $u_0$ satisfies
\begin{align*}
& \partial_t u_0 + \mathcal A^0_t u_0 = 0,\quad u_{ 0}(T_k, x) = \tilde g(x)
\end{align*}
with
\begin{align}
\label{gen_u0}
\mathcal A^0_t u_0(t,x)  &  = \sum_{i=1}^n b^i_0(t,x) x_i \frac{\partial
  u_0(t, x)}{\partial x_i}  + \frac{1}{2}\sum_{i,j=1}^n \Sigma_{i j }(t)  x_i x_j \frac{\partial^2 u_0(t, x)}{\partial
  x_i \partial x_j}\\
  \label{e:drift-b0}
b_0^i(t,x)& = - \sum_{j=i+1}^n \Sigma_{i j }(t)  \frac{\delta_j
  x_j}{1+ \delta_j x_j}. 
\end{align}
Hence, by the Feynman-Kac formula
\begin{align}
\label{eq:u0}
u_0(t, x)&  = E^{\Q^{T_n}}\left[ \tilde g(X^{t,x}_{T_k}) \right]
\end{align}
where the process $X^{t, x} = (X^{i,t,x})_{i=1}^n$ satisfies the stochastic differential equation
\begin{align}
\label{eq:LMM-diffusion}
d X^{i,t,x}_s & =  X^{i,t,x}_s\{b^i_0(s, X^{i,t,x}_s) + \sigma_i d W_s \},\quad X^{i,t,x}_t = x_i,
\end{align}
with $ W$ a $d$-dimensional standard Brownian motion with respect to $ \Q^{T_n} $ and $\sigma$ an $n \times d$-dimensional matrix such that $\sigma \sigma^{\top} = (\Sigma_{i,j})_{i,j=1}^{n}$. %and the drift $b$ is defined by

To obtain an explicit approximation for the higher order terms $u_1(t, x)$  and $u_2(t, x)$ given above, we consider the following proposition. 
\begin{proposition}
\label{p:martingale}
Let $Y$ be an $n$-dimensional log-normal process whose components follow the dynamics
$$
dY^i_t = Y^i_t (\mu_i(t) dt + \sigma_i(t) dW_t),
$$
where $\mu$ and $\sigma$ are measurable functions such that 
$$
\int_0^T (\|\mu(t)\| + \|\sigma(t)\|^2)dt <\infty
$$
and for all $y\in \mathbb R^n$ and some $\varepsilon>0$,  
$$\inf_{0\leq t\leq T}y \sigma(t) \sigma(t)^T y^T\geq \varepsilon \|y\|^2.
$$ 
We denote by $Y^{t,y}$ the process starting from $y$ at time $t$, and by $Y^{t,y,i}$ the $i$-th component of this process. Let $f$ be a bounded measurable
function and define
$$
v(t,y) = \mathbb E[f(Y^{t,y}_T)].
$$ Then, for all $i_1,\dots,i_m$, the process
$$
Y^{t,y,i_1}_s \dots Y^{t,y,i_m}_s \frac{\partial^m v (Y^{t,y}_s)}{\partial
  y_{i_1}\dots \partial y_{i_m}} ,\quad s\geq t,
$$
is a martingale. 
\end{proposition}
The proof can be carried out by direct differentiation for smooth $f$
together with a standard approximation argument for a general measurable $f$.  \\

Furthermore, we assume the following simplification for the drift terms:
\begin{itemize}
\item[] For all $i=1, \ldots, n-1$ and $p=1, \ldots, n-k-1$, the random quantities in the terms $b^i_p(t, L_t)$  in the expansion of the drift  of the Libor rates under
  the terminal measure are constant and equal to their value at time
  $t$, i.e. for all $j=1, \ldots, n$:
  \begin{align}
  \label{eq:frozen-drift}
  \frac{\delta_{j}
  L^{j}_s}{1+ \delta_{j} L^{j}_s} = \frac{\delta_{j}
  L^{j}_t}{1+ \delta_{j} L^{j}_t},  \quad \textrm{for all} \ \ s\geq t.
  \end{align}
 This simplification is known as \emph{freezing of the drift} and is often used  for pricing in the Libor market models. 
\end{itemize}

Coming back now to the first-order term $u_1$, we see that it is the solution of 
\begin{align}
\label{eq:PIDE-u1}
& \partial_t u_1 + \mathcal A^0_t u_1 + \mathcal A^1_t u_0 = 0,\quad u_{1}(T_{k},x) = 0
\end{align}
with
\begin{align}
\label{gen_u1}
\mathcal A^1_t u_0 (t, x) & = \sum_{j=1}^n b^j_1(t,x) x_j \frac{\partial  u_0(t, x)}{\partial x_j}  \\
\notag & \quad \qquad  + \frac{1}{6}  \sum_{i_1,i_2, i_3 =  1}^{n} x_{i_1} x_{i_2} x_{i_3}\frac{\partial^3  u_0(t, x)}{\partial  x_{i_1} \partial  x_{i_2} \partial x_{i_3}} M_t^3(\lambda^{i_1}, \lambda^{i_2}, \lambda^{i_3}) 
\end{align}
and the drift term 
\begin{align}
\label{eq:drift-order-1}
& b_1^j(t,x) = - \sum_{j_0=j+1}^n \sum_{j_1 =
  j_0+1}^n M_t^3 (\lambda^{j}, \lambda^{j_0}, \lambda^{j_1}) \frac{\delta_{j_0}
  x_{j_0}}{1+ \delta_{j_0} x_{j_0}}\frac{\delta_{j_1}
  x_{j_1}}{1+ \delta_{j_1} x_{j_1}}.
  \end{align}

Moreover,
\begin{align*}
 \mathcal{A}^0_t u_1(t, x) & = \sum_{i=1}^n b^i_0(t,x) x_i \frac{\partial u_1(t, x)}{\partial x_i}  + \frac{1}{2}\sum_{i,j=1}^n \Sigma_{ij}(t) x_i x_j \frac{\partial^2 u_1(t, x)}{\partial  x_i \partial x_j}.
\end{align*}
We have
\begin{lemma}
Consider the model \eqref{eq:alpha-Libor}. Under the simplification \eqref{eq:frozen-drift}, the first-order term $u_1(t, x)$ in the expansion  \eqref{eq:expansion-u} can be approximated  by 
\begin{align}
\label{eq:u1-approx}
\notag u_1(t,x) &\approx  
 \frac{1}{6}  \sum_{i_1,i_2, i_3 =1}^n x_{i_1} x_{i_2} x_{i_3}
\frac{\partial^3 u_0(t,x)}{\partial x_{i_1}\partial x_{i_2}\partial
  x_{i_3}}\int_t^{T_k}
M^3_s(\lambda^{i_1},\lambda^{i_2},\lambda^{i_3}) ds\\
\notag & -  \sum_{j=1}^n \sum_{j_0=j+1}^n \sum_{j_1=j_0+1}^n    \frac{\delta_{j_0}
  x_{j_0}}{1+ \delta_{j_0} x_{j_0}}  \frac{\delta_{j_1}
  x_{j_1}}{1+ \delta_{j_1} x_{j_1}} 
  x_j \frac{\partial u_0(t,x)}{\partial x_j} \int_t^{T_k}
M^3_s(\lambda^{j},\lambda^{j_0},\lambda^{j_1}) ds \\
  & =: \tilde u_1(t, x).
\end{align}
\end{lemma}

\begin{proof}
Applying the Feynman-Kac formula to \eqref{eq:PIDE-u1}, we have,
\begin{align}
\label{eq:u1}
\notag u_1(t,x) &= \frac{1}{6}\int_t^{T_{k}}ds \sum_{i_1,i_2, i_3 =1}^n
M^3_s(\lambda^{i_1}, \lambda^{i_2}, \lambda^{i_3})\mathbb E^{\Q^{T_n}}\left[  X^{t,x,i_1}_s
  X^{t,x,i_2}_s  X^{t,x,i_3}_s
\frac{\partial^3 u_0(s,X^{t,x}_s)}{\partial x_{i_1}\partial x_{i_2}\partial
  x_{i_3}}\right]\\
& + \int_t^{T_{k}} ds \sum_{j=1}^n \mathbb E^{\Q^{T_n}}\left[b^j_1(s,X^{t,x}_s)
  X^{t,x,j}_s \frac{\partial u_0(s,X^{t,x}_s)}{\partial x_j}\right],
\end{align}
with the  process $(X^{t,x}_s)$ defined by \eqref{eq:LMM-diffusion}. 
Under the simplification \eqref{eq:frozen-drift}, we can apply Proposition \ref{p:martingale} to obtain \eqref{eq:u1-approx}.
\end{proof}

Similarly, the second-order term $u_2$ is the solution of 
\begin{align}
\label{eq:PIDE-u2}
& \partial_t u_2 + \mathcal A^0_t u_2 + \mathcal A^1_t u_1 +  \mathcal A^2_t u_0  = 0,\quad u_{2}(T_{k},x) = 0
\end{align}
with
\begin{align}
\label{gen_u2}
\mathcal A^2_t u_0 (t, x) & = \sum_{j=1}^n b^j_2(t,x) x_j \frac{\partial  u_0(t, x)}{\partial x_j}  \\
\notag &  + \frac{1}{24}  \sum_{i_1,i_2, i_3, i_4 =  1}^{n} x_{i_1} x_{i_2} x_{i_3}  x_{i_4}\frac{\partial^4  u_0(t, x)}{\partial  x_{i_1} \partial  x_{i_2} \partial x_{i_3}  x_{i_4}} M_t^4(\lambda^{i_1}, \lambda^{i_2}, \lambda^{i_3},  \lambda^{i_4}) 
\end{align}
and the drift 
\begin{align}
\label{eq:drift-order-2}
\notag  & b_2^j(t,x) = - \sum_{j_0=j+1}^n \sum_{j_1 =
  j_0+1}^n \sum_{j_2 =
  j_1+1}^n M_t^4 (\lambda^{j}, \lambda^{j_0}, \lambda^{j_1}, \lambda^{j_2}) \frac{\delta_{j_0}
  x_{j_0}}{1+ \delta_{j_0} x_{j_0}}  \\
  & \quad \quad \quad \quad \quad \quad \cdot  \frac{\delta_{j_1}
  x_{j_1}}{1+ \delta_{j_1} x_{j_1}} \frac{\delta_{j_2}
  x_{j_2}}{1+ \delta_{j_2} x_{j_2}}.
  \end{align}
  
  \begin{lemma}
Consider the model \eqref{eq:alpha-Libor}. Under the simplification \eqref{eq:frozen-drift}, the second-order term $u_2(t, x)$ in the expansion  \eqref{eq:expansion-u} can be approximated  by 
\begin{align}
\label{eq:u2-approx}
u_2(t,x) &\approx  \tilde u_2(t, x) := \tilde E_1 + \tilde E_2 +  \tilde E_3 +  \tilde E_4,
\end{align}
with
\begin{align}
\label{eq:E1}
\notag  \tilde E_1 & :=  \frac{1}{6} \sum_{i_1, i_2, i_3 =1}^n x_{i_1} x_{i_2} x_{i_3} \int_t^{T_{k}} ds 
M^3_s(\lambda^{i_1}, \lambda^{i_2}, \lambda^{i_3}) \\ 
& \cdot \left[ \frac{1}{6} \sum_{i_4, i_5, i_6 =1}^n \left( \int_s^{T_{k}} M^3_v(\lambda^{i_4}, \lambda^{i_5}, \lambda^{i_6}) dv \right)  \right. 
  \frac{\partial^3 v^{i_4, i_5, i_6}(t, x)}{\partial x_{i_1}\partial x_{i_2}\partial x_{i_3}} \\
\notag & \left. \quad \quad -  \sum_{ j_4=1}^n  \sum_{j_5 = j_4 +1}^n \sum_{j_6 = j_5 +1}^n  
\left( \int_s^{T_{k}} M^3_v(\lambda^{j_4}, \lambda^{j_5}, \lambda^{j_6}) dv \right) 
  \frac{\partial^3 \bar v^{j_4, j_5, j_6}(t, x)}{\partial x_{i_1}\partial x_{i_2}\partial x_{i_3}} \right]
\end{align}
\begin{align}
\label{eq:E2}
 \notag  \tilde E_2 & := -  \sum_{j = 1}^n \sum_{j_0 = j +1}^n \sum_{j_1 = j_0 +1}^n \frac{\delta_{j_0} x_{j_0}}{1+ \delta_{j_0} x_{j_0}}  \frac{\delta_{j_1} x_{j_1}}{1+ \delta_{j_1} x_{j_1}} x_{j}  
 \int_t^{T_{k}} ds M_s(\lambda^{j}, \lambda^{j_0}, \lambda^{j_1})  \\
& \cdot \left[ \frac{1}{6} \sum_{i_4, i_5, i_6 =1}^n \left( \int_s^{T_{k}} M^3_v(\lambda^{i_4}, \lambda^{i_5}, \lambda^{i_6}) dv \right) \right.   \frac{\partial v^{i_4, i_5, i_6}(t, x)}{\partial x_{j}} \\ 
\notag & \left. - \sum_{ j_4 =1}^{n}  \sum_{j_5 = j_4 +1}^n \sum_{j_6 = j_5 +1}^n 
\left( \int_s^{T_{k}} M^3_v(\lambda^{j_4}, \lambda^{j_5}, \lambda^{j_6}) dv \right) 
 \frac{\partial^3 \bar v^{j_4, j_5, j_6}(t, x)}{\partial x_{j}} \right]
\end{align}
\begin{align}
\label{eq:E3}
 \tilde E_3 & := \frac{1}{24}  \sum_{i_1,i_2, i_3, i_4 =1}^n x_{i_1} x_{i_2} x_{i_3} x_{i_4}
\frac{\partial^4 u_0(t,x)}{\partial x_{i_1}\partial x_{i_2}\partial
  x_{i_3}\partial x_{i_4}}\int_t^{T_k}ds 
M^4_s(\lambda^{i_1},\lambda^{i_2},\lambda^{i_3}, \lambda^{i_4})
\end{align}
and
\begin{align}
\label{eq:E4}
 \notag  \tilde E_4 & := 
 -  \sum_{j=1}^n \sum_{j_0=j+1}^n \sum_{j_1=j_0+1}^n   \sum_{j_2=j_1+1}^n   \frac{\delta_{j_0}
  x_{j_0}}{1+ \delta_{j_0} x_{j_0}}  \frac{\delta_{j_1}
  x_{j_1}}{1+ \delta_{j_1} x_{j_1}} \frac{\delta_{j_2}
  x_{j_2}}{1+ \delta_{j_2} x_{j_2}} 
  x_j \frac{\partial u_0(t,x)}{\partial x_j}  \\
  & \quad \quad \quad  \cdot \int_t^{T_k}
M^4_s(\lambda^{j},\lambda^{j_0},\lambda^{j_1}, \lambda^{j_2}) ds
\end{align}
where we define 
\begin{align}
\label{eq:v}
v^{i, j, l}(t, x) & := x_i x_j x_l \frac{\partial^3 u_0(t, x)}{\partial x_{i}\partial x_{j}\partial x_{l}}  
\end{align}
for all $i, j, l =1, \ldots, n$
and 
\begin{align}
\label{eq:v-bar}
\bar v^{i, j, l}(t, x) & := x_i \frac{\delta_j x_j}{1+ \delta_j x_j}  \frac{\delta_l x_l}{1+ \delta_l x_l} \frac{\partial u_0(t, x)}{\partial x_{i}} 
\end{align}
for all $i =1, \ldots, n$, $j=i+1, \ldots, n$ and  $l=j+1, \ldots, n$. 
\end{lemma}

\begin{proof}
Once again by the Feynman-Kac formula applied to \eqref{eq:PIDE-u2} we have 
\begin{align}
\label{eq:u2}
\notag u_2(t,x) &= \frac{1}{6}\int_t^{T_{k}}ds \sum_{i_1,i_2, i_3 =1}^n
M^3_s(\lambda^{i_1}, \lambda^{i_2}, \lambda^{i_3})\mathbb E^{\Q^{T_n}}\left[  X^{t,x,i_1}_s
  X^{t,x,i_2}_s  X^{t,x,i_3}_s
\frac{\partial^3 u_1(s,X^{t,x}_s)}{\partial x_{i_1}\partial x_{i_2}\partial
  x_{i_3}}\right]\\
\notag & + \int_t^{T_{k}} ds \sum_{j=1}^n \mathbb E^{\Q^{T_n}}\left[b^j_1(s,X^{t,x}_s)
  X^{t,x,j}_s \frac{\partial u_1(s,X^{t,x}_s)}{\partial x_j}\right] \\
  & + \frac{1}{24}\int_t^{T_{k}}ds \sum_{i_1,i_2, i_3, i_4 =1}^n
M^4_s(\lambda^{i_1}, \lambda^{i_2}, \lambda^{i_3}, \lambda^{i_4}) \\ 
\notag & \quad \quad \quad  \cdot \mathbb E^{\Q^{T_n}}\left[  X^{t,x,i_1}_s
  X^{t,x,i_2}_s  X^{t,x,i_3}_s X^{t,x,i_4}_s
\frac{\partial^4 u_0(s,X^{t,x}_s)}{\partial x_{i_1}\partial x_{i_2}\partial
  x_{i_3} \partial x_{i_4} }\right]\\
\notag   & + \int_t^{T_{k}} ds \sum_{j=1}^n \mathbb E^{\Q^{T_n}}\left[b^j_2(s,X^{t,x}_s)
  X^{t,x,j}_s \frac{\partial u_0(s,X^{t,x}_s)}{\partial x_j}\right]  \\
  \notag & =: E_1 + E_2 +  E_3 +  E_4
\end{align}
with the  process $(X^{t,x}_s)$ given by \eqref{eq:LMM-diffusion}, $b^j_1(s, x)$ by \eqref{eq:drift-order-1} and  $b^j_2(s, x)$  by \eqref{eq:drift-order-2}.

In order to obtain an explicit expression for $u_2(t, x)$, we apply Proposition \ref{p:martingale} combined with the simplification \eqref{eq:frozen-drift} for the drift terms $b_1^j$ and $b_2^j$ above. More precisely, the expressions for the third and the fourth expectation, which are present in the terms  $E_3 $ and $E_4 $, follow by a straightforward application of Proposition \ref{p:martingale} after using the simplification for $b_2^j$. We get 
\begin{align*}
E_3 & \approx  \tilde E_3  \qquad \textrm{and} \qquad E_4  \approx  \tilde E_4
\end{align*}
with $ \tilde E_3$ and $ \tilde E_4$ given by   \eqref{eq:E3} and  \eqref{eq:E4}, respectively.

 To obtain explicit expressions for $E_1 $ and $E_2$, firstly we  insert the expression for $u_1(s, X_s^{t, x})$ as given by \eqref{eq:u1}. After some straightforward calculations, based again on the application of  Proposition \ref{p:martingale} and the simplification   \eqref{eq:frozen-drift} for $b_1^j$, which yields 
 \begin{align*}
E_1 & \approx  \tilde E_1  \qquad \textrm{and} \qquad E_2 \approx  \tilde E_2
\end{align*}
with $ \tilde E_1$ and $ \tilde E_2$ given by  \eqref{eq:E1} and  \eqref{eq:E2}, respectively.

Collecting the terms above concludes the proof. %yields \eqref{eq:u2-approx}.
\end{proof}

Summarizing, we get the following expansion for the time-$t$ price $P^{\alpha}(t; g)$ of the payoff $g(L_{T_k})$ when $\alpha \to 0$. 

\begin{proposition}
\label{p:general-payoff}
Consider the model \eqref{eq:alpha-Libor} and a European-type payoff with maturity $T_k$ given by $\xi=g(L_{T_k})$. Assuming \eqref{eq:frozen-drift}, its time-$t$ price $P^{\alpha}(t; g)$ for $\alpha \to 0$ satisfies 
\begin{align}
P^{\alpha}(t; g) & = P_0(t; g) + \alpha P_1(t; g)  + \alpha^2 P_2(t; g)  + O(\alpha^3),
\end{align} 
with 
\begin{align*}
P_0(t; g) &  := B_t(T_n) u_0(t, L_t)=: P^{LMM}(t; g) \\ \\
\notag P_1(t; g)  &: =  B_t(T_n)  u_1(t, L_t) \approx B_t(T_n)  \tilde u_1(t, L_t) \\ \\ 
 \notag P_2(t; g)  & := B_t(T_n)  u_2(t, L_t) \approx   B_t(T_n)  \tilde u_2(t, L_t)
\end{align*}
where $P^{LMM}(t; g) $ denotes the time-$t$ price of the payoff $g(L_{T_k})$ in the log-normal LMM with covariance matrix $\Sigma$ and the drift  given by \eqref{e:drift-b0}, $M^3_s(\lambda^{i_1},\lambda^{i_2},\lambda^{i_3}) $ and $M^3_s(\lambda^{j},\lambda^{j_0},\lambda^{j_1})$ are given by \eqref{eq:moments-Levy-measure},  $u_0(t, x)$ by \eqref{eq:u0} and $\tilde u_1(t,x)$  and $\tilde u_2(t,x)$ by  \eqref{eq:u1-approx} and \eqref{eq:u2-approx}, respectively.
\end{proposition}

\subsection{Approximate pricing of caplets}

Recalling that the caplet price is given by \eqref{eq:caplet-price}, where $u$ is the solution of the PIDE \eqref{eq:caplet-PIDE}, we can approximate this price using the development 
$$
u^{\alpha}(t, x) = u_0(t, x) + \alpha u_1(t, x) + \alpha^2 u_2(t, x)  + O(\alpha^3)
$$
where the zero-order term $u_0$ satisfies
\begin{align*}
&\partial_t u_0 + \mathcal A^0_t u_0 = 0,\quad u_{ 0}(T_{k-1},x) = (x_k - K)^+ \prod_{j=k+1}^{n} (1 + \delta_j x_j)
\\&\text{with}\quad \mathcal
A^0_t u_0 = \sum_{i=1}^n b^i_0(t,x) x_i \frac{\partial
  u_0(t, x)}{\partial x_i}  + \frac{1}{2}\sum_{i,j=1}^n \Sigma_{ij}(t) x_i x_j \frac{\partial^2 u_0(t, x)}{\partial
  x_i \partial x_j}\\
&\text{and}\quad b_0^i(t,x) = - \sum_{j=i+1}^n \Sigma_{ij}(t) \frac{\delta_j
  x_j}{1+ \delta_j x_j}. 
\end{align*}

The solution to the above PDE can be found via the Feynman-Kac
formula, where the conditional expectation is computed in the log-normal LMM
model with covariation matrix $(\Sigma_{i j})_{i, j=1}^n$ as in Section \ref{asympt-general-payoff}. Performing
a measure change from $\Q^{T_n}$ to $\Q^{T_k} $ and denoting by
$P_{BS}(V, S, K)$ the Black-Scholes price of a call option with
variance $V$,
$$
P_{BS}(V, S, K) = \E \left[ \left( S e^{-\frac{V}{2} + \sqrt{V}Z} - K
  \right)^+\right],\quad Z\sim N(0,1), 
$$ 
we see that the zero-order term is given by
\begin{align}
\label{eq:caplet-zero-order}
u_0(t, x) & = P_{BS}(V^{Cpl}_{t,T}, x_k, K)  \prod_{j=k+1}^{n} (1 + \delta_j x_j), 
\end{align}
where 
\begin{align}
\label{eq:caplet-variance}
V^{Cpl}_{t,T}& := \int_t^T  \Sigma_{kk}(s) ds.
\end{align}

Now, in complete analogy to the case of a general payoff, the first-order term $u_1(t, x)$ and the second-order term $u_2(t, x)$ are given by \eqref{eq:u1} and  \eqref{eq:u2}, respectively, with $u_0(t, x)$ as in \eqref{eq:caplet-zero-order}. Noting that $u_0(t, x)$ depends only on $x_k, x_{k+1}, \ldots, x_n$, the derivatives of $u_0(t, x)$ with respect to $x_1, \ldots, x_{k-1}$ are zero and the sums in \eqref{eq:u1} and  \eqref{eq:u2} in fact start from the index $k$. An application of Proposition \ref{p:martingale} and simplification \eqref{eq:frozen-drift} thus yields 
the following proposition, which provides an approximation of the caplet price $P_t^{Cpl, \alpha}(T_k, K)$ when $\alpha \to 0$. 

\begin{proposition}
Consider the model \eqref{eq:alpha-Libor} and a caplet with  strike $K$ and maturity $T_{k-1}$. Assuming \eqref{eq:frozen-drift}, its time-$t$ price $P^{Cpl, \alpha}(t; T_{k-1}, T_k,  K)$ for $\alpha \to 0$ satisfies
\begin{align}
P^{Cpl, \alpha}(t; T_{k-1}, T_k,  K) & = P^{Cpl}_0(t; T_{k-1}, T_k,  K) + \alpha P^{Cpl}_1(t; T_{k-1}, T_k,  K)  \\
\notag & + \alpha^2 P^{Cpl}_2(t; T_{k-1}, T_k,  K)  + O(\alpha^3),
\end{align} 
with 
\begin{align*}
 P^{Cpl}_0(t; T_{k-1}, T_k,  K)  &  := B_t(T_n) \delta_k u_0(t, L_t) \\
& = B_t(T_n) \delta_k P_{BS}(V^{Cpl}_{t,T_{k-1}}, L_t^k, K) \prod_{j=k+1}^{n} (1 + \delta_j L_t^j)  \quad \quad \quad
\end{align*}
\begin{align*}
 \notag & P^{Cpl}_1(t; T_{k-1}, T_k,  K) \\
& :=  B_t(T_n)  \delta_k \left\{\frac{1}{6}  \sum_{i_1,i_2, i_3 =k}^n L_t^{i_1} L_t^{i_2} L_t^{i_3}
\frac{\partial^3 u_0(t,x)}{\partial x_{i_1}\partial x_{i_2}\partial
  x_{i_3}}\Big|_{x=L_t} \int_t^{T_{k-1}}
M^3_s(\lambda^{i_1},\lambda^{i_2},\lambda^{i_3}) ds \right. \\
&\quad \quad  -  \sum_{j=k}^n \sum_{j_0=j+1}^n \sum_{j_1=j_0+1}^n    \frac{\delta_{j_0}
  L_t^{j_0}}{1+ \delta_{j_0} L_t^{j_0}}  \frac{\delta_{j_1}
  L_t^{j_1}}{1+ \delta_{j_1} L_t^{j_1}} 
  L_t^j \frac{\partial u_0(t,x)}{\partial x_j}\Big|_{x=L_t} \\
  & \left. \quad \quad \quad \quad \quad \cdot \int_t^{T_{k-1}}
M^3_s(\lambda^{j},\lambda^{j_0},\lambda^{j_1})  ds  \right\}
\end{align*}
\begin{align*}
\notag&  P^{Cpl}_2(t; T_{k-1}, T_k,  K)     \\
& := B_t(T_n)  \delta_k \left\{ \frac{1}{6} \sum_{i_1, i_2, i_3 =k}^n L_t^{i_1} L_t^{i_2} L_t^{i_3} \int_t^{T_{k-1}} ds 
M^3_s(\lambda^{i_1}, \lambda^{i_2}, \lambda^{i_3}) \right. \\
& \quad \quad  \cdot \left[ \frac{1}{6} \sum_{i_4, i_5, i_6 =k}^n \left( \int_s^{T_{k-1}} M^3_v(\lambda^{i_4}, \lambda^{i_5}, \lambda^{i_6}) dv \right)  \right.    \frac{\partial^3 v^{i_4, i_5, i_6}(t, x)}{\partial x_{i_1}\partial x_{i_2}\partial x_{i_3}} \Big|_{x=L_t} \\
& \left. \quad \quad \quad -  \sum_{ j_4=k}^n  \sum_{j_5 = j_4 +1}^n \sum_{j_6 = j_5 +1}^n  
\left( \int_s^{T_{k-1}} M^3_v(\lambda^{j_4}, \lambda^{j_5}, \lambda^{j_6}) dv \right) 
  \frac{\partial^3 \bar v^{j_4, j_5, j_6}(t, x)}{\partial x_{i_1}\partial x_{i_2}\partial x_{i_3}} \Big|_{x=L_t} \right]  
  \end{align*}
  \begin{align*}
  & -  \sum_{j = k}^n \sum_{j_0 = j +1}^n \sum_{j_1 = j_0 +1}^n \frac{\delta_{j_0} L_t^{j_0}}{(1+ \delta_{j_0} L_t^{j_0}}  \frac{\delta_{j_1} L_t^{j_1}}{(1+ \delta_{j_1} L_t^{j_1}} L_t^{j}  
 \int_t^{T_{k-1}} ds M_s(\lambda^{j}, \lambda^{j_0}, \lambda^{j_1})  \\
& \quad \quad \cdot \left[ \frac{1}{6} \sum_{i_4, i_5, i_6 =k}^n \left( \int_s^{T_{k-1}} M^3_v(\lambda^{i_4}, \lambda^{i_5}, \lambda^{i_6}) dv \right) \right.   \frac{\partial v^{i_4, i_5, i_6}(t, x)}{\partial x_{j}} \Big|_{x=L_t}\\ 
& \quad \quad  \quad \left. - \sum_{ j_4 =k}^{n}  \sum_{j_5 = j_4 +1}^n \sum_{j_6 = j_5 +1}^n 
\left( \int_s^{T_{k-1}} M^3_v(\lambda^{j_4}, \lambda^{j_5}, \lambda^{j_6}) dv \right) 
 \frac{\partial^3 \bar v^{j_4, j_5, j_6}(t, x)}{\partial x_{j}}\Big|_{x=L_t} \right] \\ \\
\notag  & + \frac{1}{24}  \sum_{i_1,i_2, i_3, i_4 =k}^n L_t^{i_1} L_t^{i_2} L_t^{i_3} L_t^{i_4}
\frac{\partial^4 u_0(t,x)}{\partial x_{i_1}\partial x_{i_2}\partial
  x_{i_3}\partial x_{i_4}} \Big|_{x=L_t}  
  \int_t^{T_{k-1}} M^4_s(\lambda^{i_1},\lambda^{i_2},\lambda^{i_3}, \lambda^{i_4}) ds\\ \\
\notag & -  \sum_{j=k}^n \sum_{j_0=j+1}^n \sum_{j_1=j_0+1}^n \sum_{j_2=j_1+1}^n   \frac{\delta_{j_0}
  L_t^{j_0}}{1+ \delta_{j_0} L_t^{j_0}}  \frac{\delta_{j_1}
  L_t^{j_1}}{1+ \delta_{j_1} L_t^{j_1}} \frac{\delta_{j_2}
  L_t^{j_2}}{1+ \delta_{j_2} L_t^{j_2}}
  L_t^j \frac{\partial u_0(t,x)}{\partial x_j}\Big|_{x=L_t} \\
  & \left. \quad \quad \quad \cdot \int_t^{T_{k-1}}
M^4_s(\lambda^{j},\lambda^{j_0},\lambda^{j_1}, \lambda^{j})  ds  \right\}
\end{align*}
with $V^{Cpl}_{t,T_{k-1}}$ given by \eqref{eq:caplet-variance}, $u_0(t, x)$ by  \eqref{eq:caplet-zero-order}, the terms $M_s^3(\cdot)$ and $M^4_s(\cdot)$ by \eqref{eq:moments-Levy-measure} and 
$v^{i_4, i_5, i_6}(t, x)$ and  $\bar v^{j_4, j_5, j_6}(t, x)$ by \eqref{eq:v} and  \eqref{eq:v-bar}, respectively.
\end{proposition}

\begin{remark}
Recalling that 
$$
u_0(t, x)  = P_{BS}(V^{Cpl}_{t,T}, x_k, K)  \prod_{j=k+1}^{n} (1 + \delta_j x_j)
$$
we see that the functions $v$ and $\bar v$ given by 
\begin{align*}
v^{i, j, l}(t, x) & := x_i x_j x_l \frac{\partial^3 u_0(t, x)}{\partial x_{i}\partial x_{j}\partial x_{l}}  
\end{align*}
for all $i, j, l =k, \ldots, n$
and 
\begin{align*}
\bar v^{i, j, l}(t, x) & := x_i \frac{\delta_j x_j}{1+ \delta_j x_j}  \frac{\delta_l x_l}{1+ \delta_l x_l} \frac{\partial u_0(t, x)}{\partial x_{i}} 
\end{align*}
for all $i = k, \ldots, n$, $j=i+1, \ldots, n$ and  $l=j+1, \ldots, n$, 
become in fact linear combinations of the terms which are polynomials in $x$ multiplied by derivatives of $P_{BS}(\cdot) $ up to order three.
\end{remark}

\subsection{Approximate pricing of swaptions}

Let us consider a swaption defined in Section \ref{s:swaptions}. For swaption pricing we again use the general result under the terminal measure $\Q^{T_n}$ given in Proposition \ref{p:general-payoff}. 
The price of the swaption $P^{Swn}(t; T_0, T_n, K)$ then satisfies
\begin{align*}
P^{Swn}(t; T_0, T_n, K) &= B_t(T_n) (u_0 (t,L_t) + \alpha u_1(t,L_t) + \alpha^2 u_2(t,L_t) ) +
O(\alpha^3)\\ &=:P_0^{Swn}(t; T_0, T_n, K)+\alpha P_1^{Swn}(t; T_0, T_n, K) \\
& \quad +\alpha^2 P_2^{Swn}(t; T_0, T_n, K) + O(\alpha^3),
\end{align*}
where the function $u_0$ satisfies the equation
$$
\partial_t u_0+ \mathcal A^{0}_t u_0 = 0,\qquad u_0(T_0,x) = \tilde g(x)
$$
with $\tilde g(x) = \delta_n f_n(x)^{-1}\left(\sum_{j=1}^n f_j(x)x_j - K\right)^+$. We see
that the zero-order term $P_0^{Swn}(t; T_0, T_n, K)$ corresponds to the price of the swaption in
the log-normal LMM model with volatility matrix $\Sigma(t)$.  

The function $u_0$ related to the swaption price in the
log-normal LMM is of course not known in explicit form but one can use
various approximations developed in the literature
\cite{jackel2003link,schoenmakers2005robust}. To introduce the
approximation of \cite{jackel2003link}, we compute the quadratic
variation of the log swap rate expressed as function of Libor rates: 
$$
R(t; T_0, T_n) = R(L^1_t,\dots,L^n_t) =  \frac{\sum_{j=1}^n \delta_j L^j_t \prod_{k=1}^j (1+\delta_k
  L^k_t )}{\sum_{j=1}^n \delta_j \prod_{k=1}^j (1+\delta_k
  L^k_t )}.
$$
\begin{align*}
\langle \log R(\cdot; T_0, T_n)  \rangle_T &= \int_0^T \frac{d\langle
  R(\cdot; T_0, T_n)  \rangle_t }{R(t; T_0, T_n)^2} = \int_0^T
\sum_{i,j=1}^n \frac{\partial R(L_t)}{\partial L^i} \frac{\partial
  R(L_t)}{\partial L^j} \frac{d\langle L^i, L^j\rangle_t}{R(t; T_0,
  T_n)^2}\\
& = \int_0^T
\sum_{i,j=1}^n \frac{\partial R(L_t)}{\partial L^i} \frac{\partial
  R(L_t)}{\partial L^j} \frac{L^i_t L^j_t \Sigma_{ij}(t) dt}{R(t; T_0,
  T_n)^2}
\end{align*}
The approximation of \cite{jackel2003link} consists in replacing all
stochastic processes in the above integral by their values at time
$0$; in other words, the swap rate becomes a log-normal random
variable such that $\log R(t; T_0, T_n) $ has variance 
$$
V^{swap}_T = \sum_{i,j=1}^n \frac{\partial R(L_0)}{\partial L^i} \frac{\partial
  R(L_0)}{\partial L^j} \frac{L^i_0 L^j_0 }{R(0; T_0,
  T_n)^2}\int_0^T \Sigma_{ij}(t) dt.
$$
The function $u_0(0,x)$ can then be approximated by applying the
Black-Scholes formula:
$$
u_0(0,x) \approx P_{BS}(V^{swap}_T, R(0; T_0, T_n), K).
$$

\section{Numerical examples}\label{numerics}

In this section, we test the performance of our approximation at pricing caplets on  Libor rates in the model  \eqref{eq:Libor-terminal-measure}, where $X_t$ is a unidimensional CGMY process \cite{cgmy.07}. The CGMY process is a pure jump process, so that $c=0$, with Lévy measure 
\[
F(dz) = \frac{C}{|z|^{1+Y}} \left( e^{-\lambda_- z} \indik_{\{x<0\}} + e^{-\lambda_+ z} \indik_{\{x>0\}} \right) \,dz\:.
\] 
The jumps of this process are not bouded from below but the
parameters we choose ensure that the probability of having a negative Libor
rate value is negligible. 
We choose the time grid $T_0=5$, $T_1=6$, ... $T_5=10$, the volatility
parameters $\lambda_i = 1$, $i=1,...,5$, the initial forward Libor
rates $L_0^i=0.06$, $i=1,...,5$ and the bond price for the first
maturity $B_0(T_0) = 1.06^{-5}$. The CGMY model parameters are chosen
according to four different cases described in the following
table, which also gives the standard deviation and excess kurtosis of $X_1$ for each case.
Case 1 corresponds to a Lévy process
that is close to the Brownian motion ($Y$ close to $2$ and
$\lambda_+$ and $\lambda_-$ large) and Case 4 is  a
Lévy process that is very far from Brownian motion. 

\medskip

%\begin{table}
%\label{CGMY-cases}
\centerline{
\begin{tabular}{l|llllcc}
Case & $C$ & $\lambda_+$ & $\lambda_-$&  $Y$ & Volatility &
Excess kurtosis\\\hline
1 & 0.01  & 10  & 20  & 1.8 & 23.2\% & 0.028\\
2 & 0.1  & 10  &  20  &  1.2 & 17\% &  0.36\\
3 & 0.2  & 10  & 20  &  0.5 &  8.7\% & 3.97\\
4 & 0.2  &  3  &  5  &  0.2 & 18.9\% & 12.7
\end{tabular}
}
%\caption{CGMY process: standard deviation and excess kurtosis for four different parameter sets}
%\end{table}
\medskip

We first calculate the price of the ATM caplet with maturity $T_1$ written on the Libor
rate $L^1$ with the zero-order, first-order and second-order
approximation, using as benchmark the jump-adapted Euler scheme of
\citeN{kohatsu2010jump}. The first Libor rate is chosen to maximize the
nonlinear effects related to the drift of the Libor rates, since the
first maturity is the farthest from the terminal date. The results are
shown in Table \ref{atm.tab}. We see that for all four cases, the
price computed by second-order approximation is within or at the
boundary of the Monte Carlo confidence interval, which is itself quite
narrow (computed with $10^6$ trajectories).

Secondly, we evaluate the prices of caplets with strikes ranging from
$3\%$ to $9\%$ and explore the performance of our analytic
approximation for estimating the caplet implied volatility smile. The
results are shown in Figure \ref{smile.fig}. We see that in cases 1,  2
and 3, which correspond to the parameter values most relevant in
practice given the value of the excess kurtosis, the second order
approximation reproduces the volatility smile quite well (in case 1
there is actually no smile, see the scale on the $Y$ axis of the graph). In case 4,
which corresponds to very violent jumps and pronounced smile, the qualitative
shape of the smile is correctly reproduced, but the actual values are
often outside the Monte Carlo interval. This means that in this
extreme case the model is
too far from the Gaussian LMM for our approximation to be precise. We also note that the algorithm runs in $\mathcal{O}(n^6)$, for the second order approximation, due to the number of partial derivatives that one has to calculate. The algorithm may therefore run slowly, should $n$ become too large.

\begin{table}
\begin{tabular}{l|llll}
& Case 1 & Case 2 & Case 3 & Case 4\\\hline
Order 0 & 0.008684 & 0.006392 & 0.003281  & 0.007112\\
Order 1 & 0.008677& 0.006361 & 0.003241 & 0.006799\\
Order 2 & \textbf{0.008677}&\textbf{0.006351} & \textbf{0.003172}  &\textbf{0.006556}\\
MC lower bound &0.008626 &0.006306&0.003178& 0.006493 \\
MC upper bound & 0.008712 & 0.006361&0.003204 &0.006578
\end{tabular}

\medskip

\caption{Price of ATM caplet computed using the analytic approximation
  together with the 95\% confidence
  bounds computed by Monte Carlo over $10^6$ trajectories.}
\label{atm.tab}
\end{table}

\begin{figure}
\includegraphics[width=0.5\textwidth]{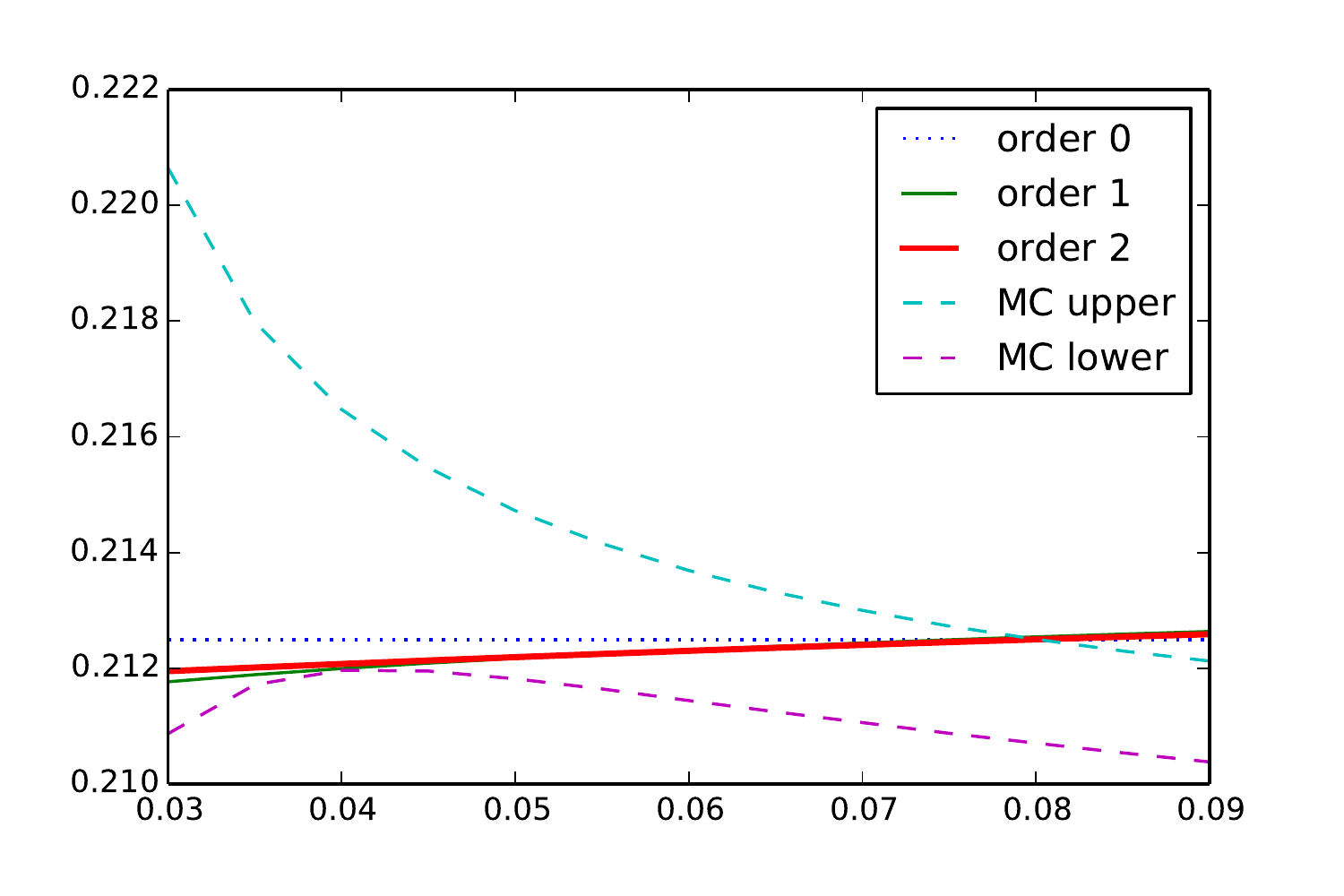}\includegraphics[width=0.5\textwidth]{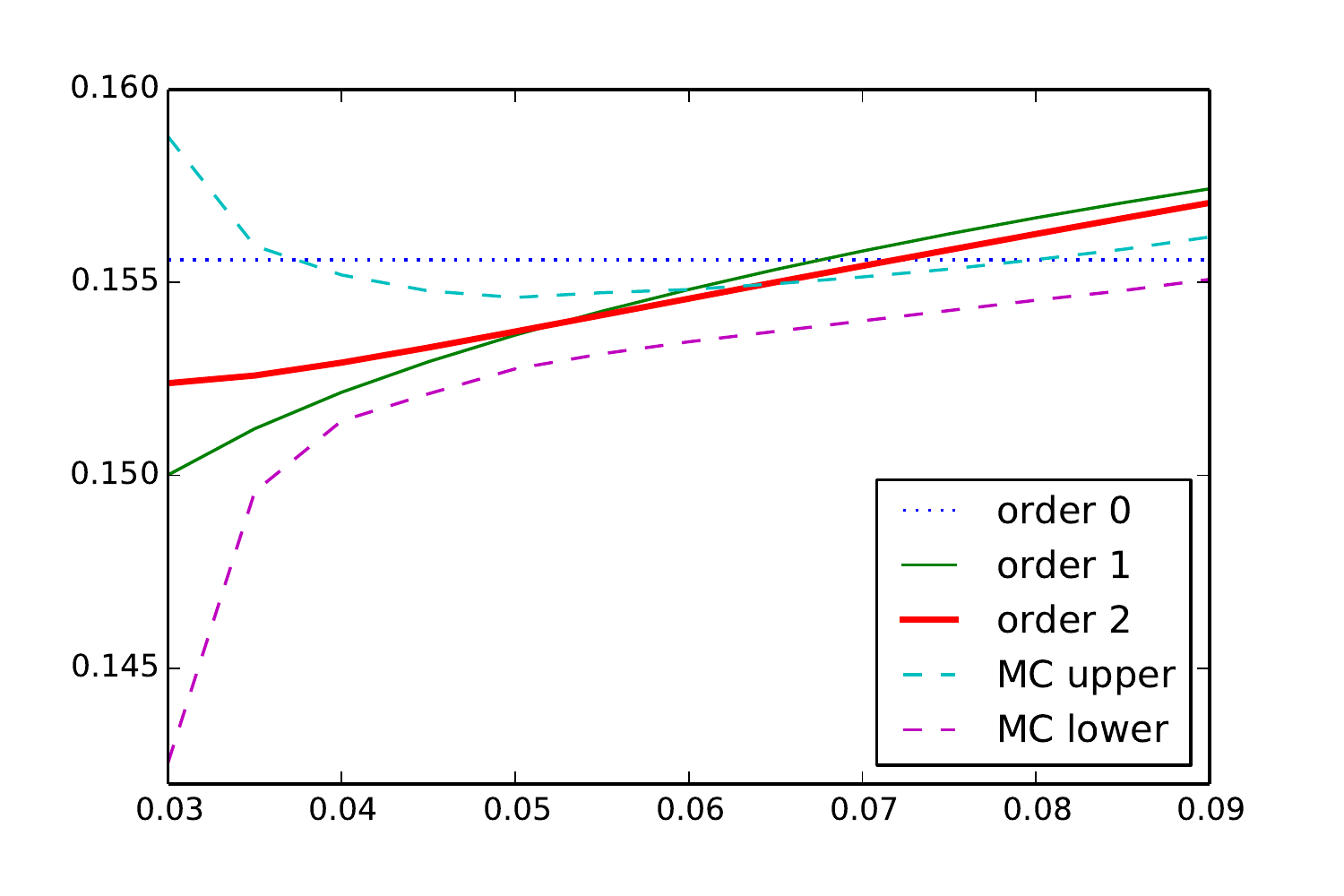}
\\
\includegraphics[width=0.5\textwidth]{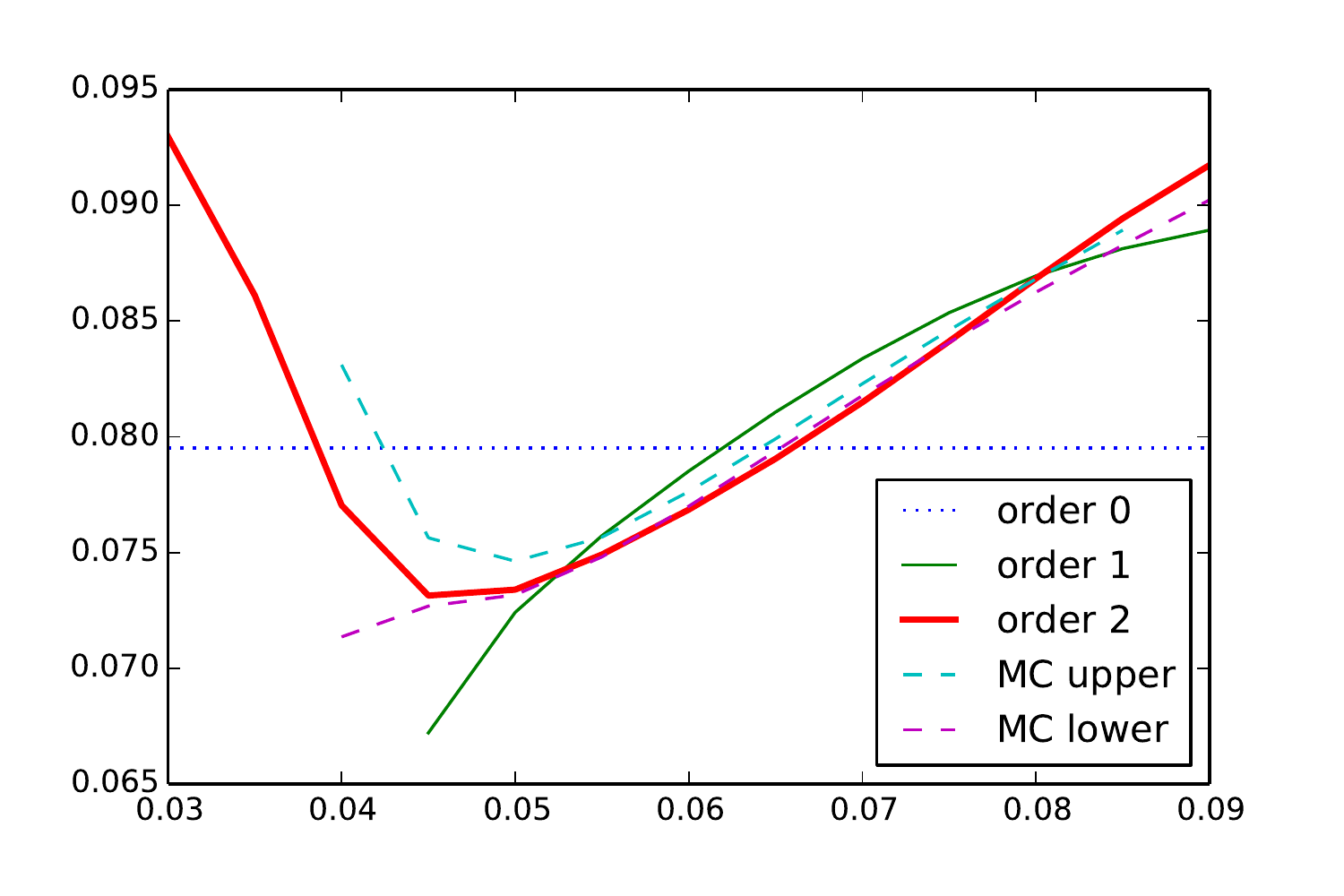}\includegraphics[width=0.5\textwidth]{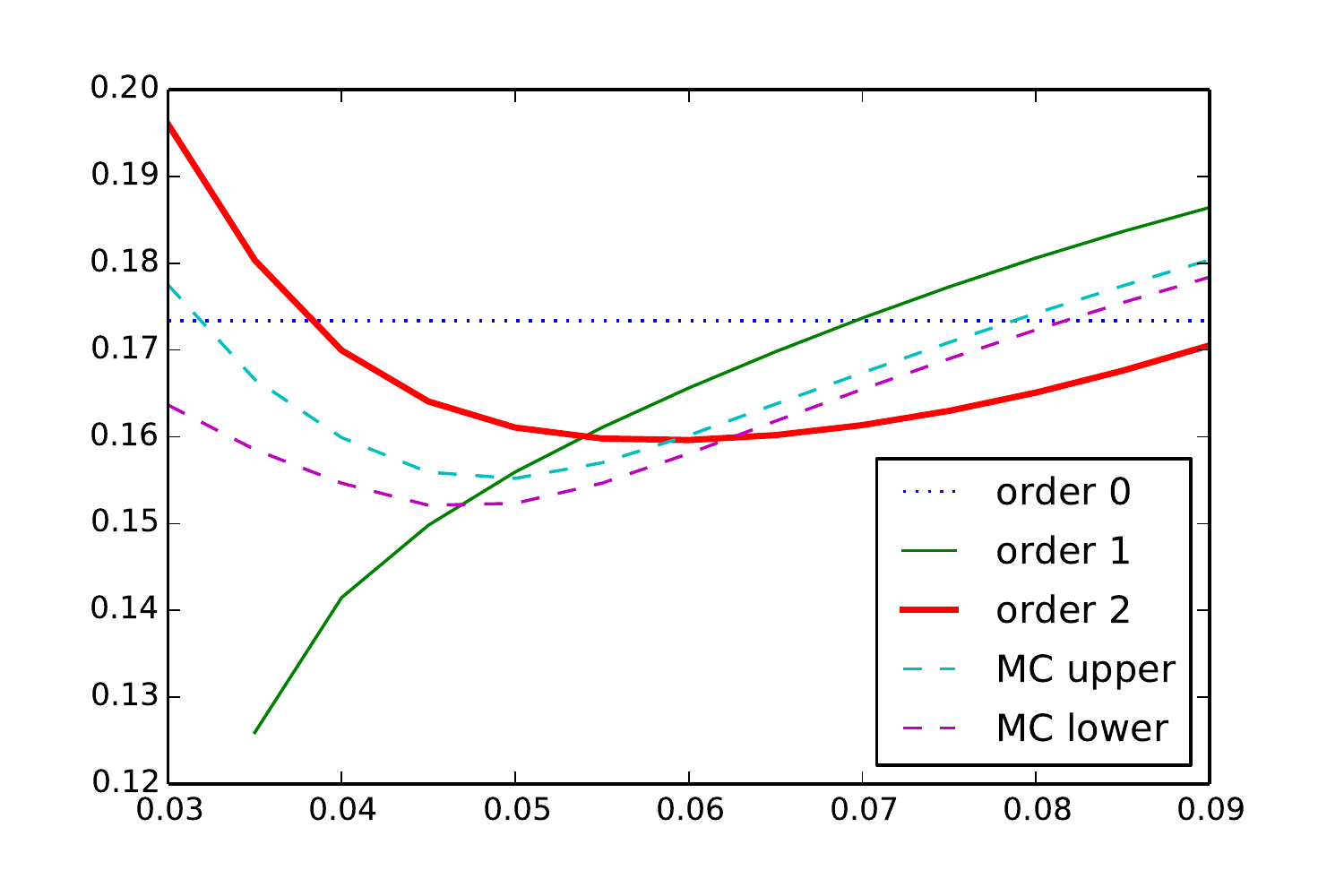}
\caption{Implied volatilities of caplets with different strikes
  computed using the analytic approximation together with the Monte
  Carlo bound. Top graphs: Case 1 (left) and Case 2 (right). Bottom
  graphs: Case 3 (left) and Case 4 (right). }
\label{smile.fig}
\end{figure}

\section{Acknowledgement} This research of Peter Tankov is partly supported
by the Chair Financial Risks of the Risk Foundation sponsored by
Société Générale.


\begin{thebibliography}{}

\bibitem[\protect\citeauthoryear{Belomestny and Schoenmakers}{Belomestny and
  Schoenmakers}{2011}]{belomestny2011jump}
Belomestny, D. and J.~Schoenmakers (2011).
\newblock A jump-diffusion libor model and its robust calibration.
\newblock {\em Quantitative Finance\/}~{\em 11\/}(4), 529--546.

\bibitem[\protect\citeauthoryear{Benhamou, Gobet, and Miri}{Benhamou
  et~al.}{2009}]{benhamou2009smart}
Benhamou, E., E.~Gobet, and M.~Miri (2009).
\newblock Smart expansion and fast calibration for jump diffusions.
\newblock {\em Finance and Stochastics\/}~{\em 13\/}(4), 563--589.

\bibitem[\protect\citeauthoryear{Benhamou, Gobet, and Miri}{Benhamou
  et~al.}{2010}]{benhamou2010time}
Benhamou, E., E.~Gobet, and M.~Miri (2010).
\newblock Time dependent heston model.
\newblock {\em SIAM Journal on Financial Mathematics\/}~{\em 1\/}(1), 289--325.

\bibitem[\protect\citeauthoryear{Brace, G\c{a}tarek, and Musiela}{Brace
  et~al.}{1997}]{BraceGatarekMusiela97}
Brace, A., D.~G\c{a}tarek, and M.~Musiela (1997).
\newblock The market model of interest rate dynamics.
\newblock {\em Mathematical Finance\/}~{\em 7}, 127--155.

\bibitem[\protect\citeauthoryear{Carr, Geman, Madan, and Yor}{Carr
  et~al.}{2007}]{cgmy.07}
Carr, P., H.~Geman, D.~B. Madan, and M.~Yor (2007).
\newblock Self-decomposability and option pricing.
\newblock {\em Mathematical Finance\/}~{\em 17}, 31--57.

\bibitem[\protect\citeauthoryear{\v{C}ern\'y, Denkl, and Kallsen}{\v{C}ern\'y
  et~al.}{2013}]{CernyDenklKallsen13}
\v{C}ern\'y, A., S.~Denkl, and J.~Kallsen (2013).
\newblock {Hedging in L\'evy models and the time step equivalent of jumps}.
\newblock Preprint, \texttt{arXiv:1309.7833}.


\bibitem[\protect\citeauthoryear{De~Franco}{De~Franco}{2012}]{defranco.thesis}
De~Franco, C. (2012).
\newblock {\em Two studies in risk management: portfolio insurance under risk
  measure constraint and quadratic hedge for jump processes.}
\newblock Ph.\ D. thesis, University Paris Diderot.

\bibitem[\protect\citeauthoryear{Eberlein and Kluge}{Eberlein and
  Kluge}{2007}]{EberleinKluge07}
Eberlein, E. and W.~Kluge (2007).
\newblock {Calibration of {L}\'{e}vy term structure models}.
\newblock In M.~Fu, R.~A. Jarrow, J.-Y. Yen, and R.~J. Elliott (Eds.), {\em
  Advances in Mathematical Finance: In Honor of D. B. Madan}, pp.\  147--172.
  Birkhäuser.

\bibitem[\protect\citeauthoryear{Eberlein and \"{O}zkan}{Eberlein and
  \"{O}zkan}{2005}]{EberleinOezkan05}
Eberlein, E. and F.~\"{O}zkan (2005).
\newblock {The L\'{e}vy Libor model}.
\newblock {\em Finance and Stochastics\/}~{\em 9}, 327--348.

\bibitem[\protect\citeauthoryear{J\"{a}ckel and Rebonato}{J\"{a}ckel and
  Rebonato}{2003}]{jackel2003link}
J\"{a}ckel, P. and R.~Rebonato (2003).
\newblock The link between caplet and swaption volatilities in a
  {B}race-{G}atarek-{M}usiela/{J}amshidian framework: approximate solutions and
  empirical evidence.
\newblock {\em Journal of Computational Finance\/}~{\em 6\/}(4), 41--60.

\bibitem[\protect\citeauthoryear{Kallsen}{Kallsen}{2006}]{Kallsen06}
Kallsen, J. (2006).
\newblock A didactic note on affine stochastic volatility models.
\newblock In Y.~Kabanov, R.~Lipster, and J.~Stoyanov (Eds.), {\em From
  Stochastic Calculus to Mathematical Finance: The Shiryaev Festschrift}, pp.\
  343--368. Springer.

\bibitem[\protect\citeauthoryear{Kluge}{Kluge}{2005}]{Kluge05}
Kluge, W. (2005).
\newblock {\em {Time-Inhomogeneous L\'evy Processes in Interest Rate and Credit
  Risk Models}}.
\newblock Ph.\ D. thesis, University of Freiburg.

\bibitem[\protect\citeauthoryear{Kohatsu-Higa and Tankov}{Kohatsu-Higa and
  Tankov}{2010}]{kohatsu2010jump}
Kohatsu-Higa, A. and P.~Tankov (2010).
\newblock Jump-adapted discretization schemes for {L}{\'e}vy-driven sdes.
\newblock {\em Stochastic Processes and their Applications\/}~{\em 120\/}(11),
  2258--2285.

\bibitem[\protect\citeauthoryear{M\'enass\'e and Tankov}{M\'enass\'e and
  Tankov}{2015}]{MenasseTankov15}
M\'enass\'e, C. and P.~Tankov (2015).
\newblock {Asymptotic indifference pricing in exponential L\'evy models}.
\newblock Preprint, \texttt{arXiv:1502.03359}.

\bibitem[\protect\citeauthoryear{Miltersen, Sandmann, and Sondermann}{Miltersen
  et~al.}{1997}]{MiltersenSandmannSondermann97}
Miltersen, K.~R., K.~Sandmann, and D.~Sondermann (1997).
\newblock {Closed form solutions for term structure derivatives with log-normal
  interest rates}.
\newblock {\em The Journal of Finance\/}~{\em 52}, 409--430.

\bibitem[\protect\citeauthoryear{Papapantoleon, Schoenmakers, and
  Skovmand}{Papapantoleon et~al.}{2012}]{papapantoleon12efficient}
Papapantoleon, A., J.~Schoenmakers, and D.~Skovmand (2012).
\newblock Efficient and accurate log-{L}\'evy approximations to {L}\'evy driven
  {LIBOR} models.
\newblock {\em Journal of Computational Finance\/}~{\em 15\/}(4), 3--44.

\bibitem[\protect\citeauthoryear{Sandmann, Sondermann, and Miltersen}{Sandmann
  et~al.}{1995}]{SandmannSondermannMiltersen95}
Sandmann, K., D.~Sondermann, and K.~R. Miltersen (1995).
\newblock Closed form term structure derivatives in a {Heath--Jarrow--Morton}
  model with log-normal annually compounded interest rates.
\newblock In {\em Proceedings of the Seventh Annual European Futures Research
  Symposium Bonn}, pp.\  145--165.
\newblock Chicago Board of Trade.

\bibitem[\protect\citeauthoryear{Schoenmakers}{Schoenmakers}{2005}]{schoenmakers2005robust}
Schoenmakers, J. (2005).
\newblock {\em Robust Libor modelling and pricing of derivative products}.
\newblock CRC Press.


\end{thebibliography}
\end{document}